\newcommand{\conferenceFull}[2]{#2}

\documentclass[a4paper,UKenglish,cleveref,autoref,thm-restate,numberwithinsect]{lipics-v2021}

\title{Separability in B\"uchi VASS and Singly Non-Linear Systems of Inequalities}

\author{Pascal Baumann}{Max Planck Institute for Software Systems (MPI-SWS), Germany}{pbaumann@mpi-sws.org}{https://orcid.org/0000-0002-9371-0807}{}

\author{Eren Keskin}{TU Braunschweig, Germany}{e.keskin@tu-bs.de}{https://orcid.org/0009-0009-5621-6568}{}

\author{Roland Meyer}{TU Braunschweig, Germany}{roland.meyer@tu-bs.de}{https://orcid.org/0000-0001-8495-671X}{}

\author{Georg Zetzsche}{Max Planck Institute for Software Systems (MPI-SWS), Germany}{georg@mpi-sws.org}{https://orcid.org/0000-0002-6421-4388}{}

\authorrunning{P. Baumann, E. Keskin, R. Meyer, and G. Zetzsche} %

\Copyright{Pascal Baumann, Eren Keskin, Roland Meyer, and Georg Zetzsche} %

\ccsdesc[500]{Theory of computation~Logic}
\ccsdesc[500]{Theory of computation~Formal languages and automata theory}

\keywords{Vector addition systems, infinite words, separability, inequalities, quantifier elimination, rational, polynomials} %

\funding{\flag[3cm]{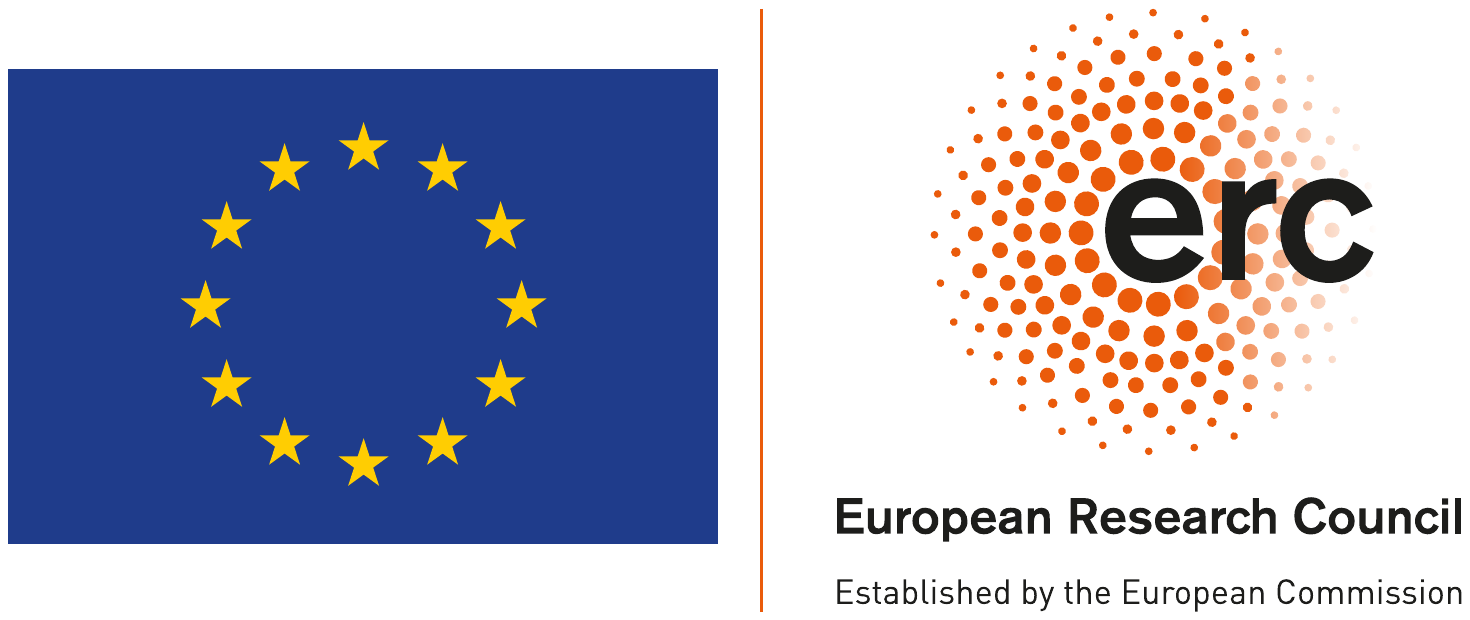}Funded by the European Union (ERC, FINABIS, 101077902). Views and opinions expressed are however those of the author(s) only and do not necessarily reflect those of the European Union or the European Research Council Executive Agency. Neither the European Union nor the granting authority can be held responsible for them.}

\nolinenumbers %

\hideLIPIcs

\EventEditors{Karl Bringmann, Martin Grohe, Gabriele Puppis, and Ola Svensson}
\EventNoEds{4}
\EventLongTitle{51st International Colloquium on Automata, Languages, and Programming (ICALP 2024)}
\EventShortTitle{ICALP 2024}
\EventAcronym{ICALP}
\EventYear{2024}
\EventDate{July 8--12, 2024}
\EventLocation{Tallinn, Estonia}
\EventLogo{}
\SeriesVolume{297}
\ArticleNo{157}

\usepackage{hyperref}
\usepackage{cite}
\usepackage{totcount}
\usepackage{newfile}

\usepackage{bm}
\usepackage{commands}

\usepackage{utfsym}

\newcommand{\cV}{\mathcal{V}}

\newcommand{\cS}{\mathcal{S}}
\newcommand{\Z}{\mathbb{Z}}
\newcommand{\N}{\mathbb{N}}
\newcommand{\Q}{\mathbb{Q}} %
\newcommand{\R}{\mathbb{R}} %
\newcommand{\ourbold}[1]{\bm{#1}}

\newcommand{\bmm}{\ourbold{m}}
\newcommand{\bmn}{\ourbold{n}}
\newcommand{\bs}{\ourbold{s}}

\newcommand{\bx}{\ourbold{x}}
\newcommand{\by}{\ourbold{y}}
\newcommand{\bz}{\ourbold{z}}
\newcommand{\bu}{\ourbold{u}}
\newcommand{\bv}{\ourbold{v}}

\newcommand{\bb}{\ourbold{b}}
\newcommand{\bc}{\ourbold{c}}

\newcommand{\bA}{\ourbold{A}}

\newcommand{\bD}{\ourbold{D}}
\newcommand{\bI}{\ourbold{I}} %
\newcommand{\bzero}{\ourbold{0}}

\newcommand{\regsep}[2]{#1\mathrel{|}#2}
\newcommand{\notregsep}[2]{#1\mathrel{\not|}#2}

\newcommand{\KM}[1]{\mathsf{KM}(#1)}
\newcommand{\realrootsof}[1]{\mathsf{Roots}_{\R}(#1)}
\newcommand{\Roots}[1]{\mathsf{Roots}(#1)}
\newcommand{\Bounds}[1]{\mathsf{Bounds}(#1)}
\newcommand{\Separation}[1]{\mathsf{Separation}(#1)}

\DeclareMathOperator\maxc{maxc} %
\DeclareMathOperator\adj{adj} %
\DeclareMathOperator\sgn{sgn} %

\DeclareDocumentCommand{\inteff}{o}{%
        \ensuremath{%
                \IfNoValueTF{#1}{%
			\delta%
                }{
			\delta_{#1}%
                }
        }
}

\DeclareDocumentCommand{\exteff}{o}{%
        \ensuremath{%
                \IfNoValueTF{#1}{%
			\varphi%
                }{
			\varphi%
                }
        }
}

\newcommand{\inteffof}[1]{\inteff(#1)}
\newcommand{\exteffof}[1]{\exteff(#1)}

\newcommand{\NP}{\mathsf{NP}}
\newcommand{\PSPACE}{\mathsf{PSPACE}}
\newcommand{\EXPSPACE}{\mathsf{EXPSPACE}}

\newcounter{inlineenum}

\usepackage[]{todonotes}

\DeclareMathOperator{\rank}{rank}

\begin{document}

\thispagestyle{empty} %

\maketitle

\begin{abstract}
  The $\omega$-regular separability problem for B\"uchi VASS coverability languages has recently been shown to be decidable, but with an $\EXPSPACE$ lower and a non-primitive recursive upper bound---the exact complexity remained open. 
  We close this gap and show that the problem is $\EXPSPACE$-complete.
  A careful analysis of our complexity bounds additionally yields a $\PSPACE$ procedure in the case of fixed dimension $\geq 1$, which matches a pre-established lower bound of $\PSPACE$ for one dimensional B\"uchi VASS.
  Our algorithm is a non-deterministic search for a witness whose size, as we show, can be suitably bounded.
  Part of the procedure is to decide the existence of runs in VASS that satisfy certain non-linear properties.
	Therefore, a key technical ingredient is to analyze a class of systems of inequalities    where one variable may occur in non-linear (polynomial) expressions.

	These so-called singly non-linear systems (SNLS) take the form $\bA(x)\cdot \by \geq \bb(x)$, where $\bA(x)$ and $\bb(x)$ are a matrix resp.\ a vector whose entries are polynomials in $x$, and $\by$ ranges over vectors in the rationals.
  Our main contribution on SNLS is an exponential upper bound on the size of rational solutions to singly non-linear systems. The proof consists of three steps. 
  First, we give a tailor-made quantifier elimination to characterize all real solutions to $x$. 
  Second, using the root separation theorem about the distance of real roots of polynomials, we show that if a rational solution exists, then there is one with at most polynomially many bits.
	Third, we insert the solution for $x$ into the SNLS, making it linear and allowing us to invoke standard solution bounds from convex geometry.

	Finally, we combine the results about SNLS with several techniques from
	the area of VASS to devise an $\EXPSPACE$ decision procedure for
	$\omega$-regular separability of B\"{u}chi VASS.
\end{abstract}

\section{Introduction}
Vector addition systems with states (VASS) are one of the most popular and well-studied models of concurrent systems. A $d$-dimensional VASS consists of finitely many control states and $d$ counters. Transitions between control states can increment or decrement the $d$ counters, but importantly, one can only take a transition if the new counter values remain non-negative.  %

\subparagraph{Separability problems}
In recent years, a strong focus of the research on VASS was on \emph{separability problems}~\cite{DBLP:conf/fsttcs/KocherZ23,Baumann23,DBLP:conf/lics/CzerwinskiZ20,DBLP:conf/fsttcs/ThinniyamZ19,DBLP:conf/icalp/CzerwinskiHZ18,DBLP:journals/lmcs/CzerwinskiL19,DBLP:conf/concur/CzerwinskiLMMKS18,DBLP:conf/icalp/ClementeCLP17,DBLP:conf/stacs/ClementeCLP17,DBLP:conf/concur/Keskin023,KeskinMeyer2024a}. Here, we label the transitions of the input VASS $\cV_1$, $\cV_2$ by letters, which gives rise to languages $L_1$ and $L_2$. Then, we ask whether there exists a language $S$, from some class $\cS$ of allowed separators, such that $L_1\subseteq S$ and $L_2\cap S=\emptyset$. Here, $\cS$ is typically the class of regular languages. 

An important motivation for studying separability problems is that
separators can be viewed as certificates for disjointness, and thus the non-existence of a run in the product of $\cV_1$ and $\cV_2$. Such certificates are crucial for understanding safety verification for infinite-state systems, %
where the difficult part is to prove the non-existence of a run (the existence of a run is usually easy to show). 
In particular, certificates for non-existence %
are often the ingredient that is conceptually hardest to come by. For example, in the case of reachability in VASS, the KLM decomposition~\cite{DBLP:journals/siamcomp/Mayr84,DBLP:conf/stoc/Kosaraju82,DBLP:journals/tcs/Lambert92,DBLP:conf/lics/LerouxS15} and Leroux's Presburger-definable inductive invariants~\cite{DBLP:conf/popl/Leroux11} can be viewed as such certificates. Regular separators could play a similar role in alternative approaches to reachability.

In addition to understanding certificates, the recent attention on separability has led to other applications. For example, work on separability by bounded languages has led to a general framework %
to address unboundedness problems for VASS~\cite{DBLP:conf/icalp/CzerwinskiHZ18}. Moreover, separability results %
were used in an algorithm for deciding inclusion between unambiguous VASS~\cite{DBLP:conf/concur/CzerwinskiH22}.

With the recent contribution by Keskin and Meyer~\cite{KeskinMeyer2024a} (together with earlier decidability results for subclasses and variants~\cite{Baumann23,DBLP:conf/lics/CzerwinskiZ20,DBLP:conf/icalp/CzerwinskiHZ18,DBLP:journals/lmcs/CzerwinskiL19,DBLP:conf/concur/CzerwinskiLMMKS18,DBLP:conf/icalp/ClementeCLP17,DBLP:conf/stacs/ClementeCLP17}), proving regular separability decidable for (finite-word) VASS, the \emph{decidability} status of regular separability has largely been settled. However, concerning \emph{complexity}, regular separability is far from understood, with few results: So far, the only exact complexity results are $\PSPACE$-completeness for (succintly represented) one-dimensional VASS~\cite{DBLP:journals/lmcs/CzerwinskiL19}, $\EXPSPACE$-completeness for VASS coverability languages~\cite{DBLP:conf/concur/CzerwinskiLMMKS18}, and Ackermann-completeness for VASS reachability languages~\cite{KeskinMeyer2024a}.

\subparagraph*{B\"{u}chi VASS} A particularly challenging problem is ($\omega$-)regular
separability in B\"{u}chi VASS~\cite{Baumann23}. In a B\"{u}chi VASS $\cV$, the
language $L(\cV)$ consists of \emph{infinite words} induced by runs that visit
some final state infinitely often. As demonstrated by Baumann, Meyer, and
Zetzsche~\cite{Baumann23}, B\"{u}chi VASS behave quite differently in terms
of regular separability from their finite-word counterpart, coverability
languages of VASS~\cite{DBLP:conf/concur/CzerwinskiLMMKS18}. Nevertheless, Baumann, Meyer, and Zetzsche
proved decidability of regular separability for B\"{u}chi VASS~\cite{Baumann23}.
However, the complexity remained open: Their algorithm requires at least
Ackermannian time (because it constructs Karp-Miller graphs),
and the only known lower bound is $\EXPSPACE$.

\subparagraph{Challenge: Non-linear constraints} Improving the complexity established in \cite{Baumann23} is challenging due to the characterization of inseparability there:
Inseparability is equivalent to the existence of a constellation of runs, called an \emph{inseparability
flower}, that must satisfy a \emph{non-linear constraint}, meaning a constraint that is not
expressible in linear arithmetic (i.e. first-order logic of $(\Z;+,<,0,1)$ or
$(\Q;+,<,0,1)$). Essentially, such a flower is a triple $(\alpha,\beta,\gamma)$
of cyclic runs such that (among other linear inequalities) the counter effect
of the combined run $\alpha\beta\gamma$ is a scalar multiple of the counter effect of just $\alpha$.
In other words, we are looking for runs with effects $\bu,\bv\in\Z^n$ such that
\begin{equation} \exists x\in\Q\colon \bv=x\cdot \bu\ .
\label{scalar-multiple}\end{equation} Detecting runs with such constraints is
difficult: There are powerful generic $\EXPSPACE$ algorithms for detecting runs
that satisfy unboundedness conditions~\cite{Demri13},
linear %
constraints~\cite{AtigH11}, or variants
of computation tree logic (CTL)~\cite{DBLP:conf/mfcs/BlockeletS11}. However, %
condition
\eqref{scalar-multiple} falls in neither of those categories. 

In fact, we are not aware of any algorithmic approach to solving systems of
linear inequalities with constraints of type \eqref{scalar-multiple} (let alone
inside algorithms for VASS). There is a result by Gurari \& Ibarra~\cite{DBLP:journals/jacm/GurariI79}
showing that integral feasibility of systems of equalities $\bA\cdot\by=\bb(x)$ can be decided in $\NP$, where $\bb(x)$ is a vector containing in each component a
quotient of polynomials in $x$. However, these do not seem to capture
\eqref{scalar-multiple}: By moving the denominators from $\bb(x)$ to the left-hand side, one obtains equations
where every variable from $\by$ is multiplied with the same polynomial over $x$. However, for \eqref{scalar-multiple}, we need to multiply a \emph{subset} of the linear variables (namely, those in $\bu$) with a polynomial (namely, $x$).
The same is
true for the logic of \emph{almost linear arithmetic} due to
Weispfenning~\cite{DBLP:journals/jsc/Weispfenning90}, whose existential fragment is also
solvable in $\NP$. Here, the definable sets are finite unions of solution sets
of Gurari \& Ibarra. 

Furthermore, it is not even clear how to detect inseparability flowers by invoking
reachability in VASS (even though this would only yield an Ackermann upper
bound): To some extent, algorithms for reachability permit non-linear
constraints---for example, using standard tricks, it is decidable whether one
can reach a configuration with counter values $(m,n)$ such that $n\le 2^m$.
However, the condition in \eqref{scalar-multiple} does not even seem to be
captured by such methods. 

\subparagraph*{Contribution} Our main result is that regular separability in
B\"{u}chi VASS is $\EXPSPACE$-complete, and $\PSPACE$-complete in fixed
dimension $\ge 1$. The key technical ingredient is a method that we expect to
be of independent interest: We develop a procedure for solving systems of
linear inequalities with a single non-linear variable, which we call
\emph{singly non-linear systems of inequalities} (SNLS). We use our results
about SNLS to show that if an inseparability witness exists, then there is one
where all runs have at most doubly exponential length, yielding an $\EXPSPACE$
procedure. In fixed dimension, we obtain singly exponential bounds, leading to a $\PSPACE$ procedure.

\subparagraph{Step I: Singly non-linear systems of inequalities}
Intuitively, a singly non-linear system of inequalities (SNLS) is a system of inequalities that is linear in all but one variable. This means, there is one variable $x$ that may appear in arbitrary polynomials, but all others can only occur linearly. More precisely, an SNLS is a system of inequalities of the form
\begin{equation} \bA(x) \cdot \by\ge \bb(x),\quad \by\ge\bzero\ , \label{snls-form} \end{equation}
where $\bA(x)\in\Z[x]^{m\times n}$ is an $m\times n$ matrix over the ring $\Z[x]$ of integer polynomials in $x$, $\bb(x)\in\Z[x]^m$ is a vector of polynomials from $\Z[x]$, and $\by$ ranges over $\Q^n$. Notice that here indeed, $x$ can be freely multiplied with itself and other variables, whereas the expression on the left-hand side must be linear in each component of $\by$. 

Our main result about SNLS is that if a system as in \eqref{snls-form} has a solution $(x,\by)\in\Q\times\Q^n$, then it has a solution where all numbers (numerators and denominators) are bounded exponentially in the description size of $\bA(x)$ and $\bb(x)$, even if numbers in the description are encoded in binary. This implies in particular that feasibility of SNLS is in $\NP$.

In the proof, we first show that the set of all $x\in\Q$ for which there is a
solution $(x,\by)$ can be described by a Boolean combination $\Phi$ of
polynomial constraints of the form $p(x)\ge 0$, for polynomials
$p\in\Z[x]$. This amounts to a quantifier elimination procedure for a class of first-order formulas in the ordered field $(\Q;+,\cdot,<,0,1)$. This is perhaps surprising, since this structure does not admit quantifier elimination in general~\cite[Theorem 2]{macintyre1983elimination}.

Let us give a geometric explanation how we arrive at the constraints $\Phi(x)$:
For each choice of $x$, the SNLS $\bA(x)\cdot\by\ge\bb(x),~\by\ge\bzero$
defines a %
polyhedron. It is a standard fact in convex geometry that
such a polyhedron has a point on a minimal face, and moreover this point can be
expressed as the inverse of a submatrix of $\bA(x)$ multiplied with $\bb(x)$.
This expression can then be plugged back into
$\bA(x)\cdot\by\ge\bb(x)$ to obtain a set of polynomial constraints on $x$,
subject to a particular determinant being non-zero. The latter non-zero
condition can as well be expressed as a polynomial constraint.

We then show that $\Phi$ has a small solution: In one case, a rational root of
one of the polynomials $p$ is a solution---these can be bounded by the
\emph{Rational Root Theorem}. The other case is that the solution $x$ lies
strictly between two roots $r_1<r_2$ of participating polynomials. But then one
can observe that any rational number between those roots is a solution (if
no other root lies between $r_1$ and $r_2$). Using the \emph{Root
Separation Theorem} (specifically, Rump's Bound~{\cite{Mishra93}}), which lower-bounds the size
of such intervals $(r_1,r_2)$, we can then conclude that such an interval must
contain a rational number with small numerator and denominator.

Once we exhibit a small $x$, we can plug it into $\bA(x)\cdot\by\ge\bb(x)$
to obtain a system of linear inequalities.
Then we use standard bounds to obtain a small (i.e.\ exponential)
solution $\by\in\Q^n$. It should be noted that while our result about SNLS
concerns rational solutions, we apply it in the case where $\bb(x)\ge \bzero$,
which means a rational solution can be turned into an integral solution by
multiplying a common denominator.

\subparagraph{Step II: Rackoff-like bounds}
After establishing the solution bound for SNLS, we use this result in the context of B\"uchi VASS to show the existence of inseparability witnesses that are small, i.e.\ consist of runs that are at
most doubly exponential in length. Here, we use an adaptation of the Rackoff
technique~\cite{DBLP:journals/tcs/Rackoff78} similar to the proofs of
Habermehl~\cite{Habermehl97} and Atig \&
Habermehl~\cite{AtigH11}. 
In
\cite{AtigH11}, it is shown that runs satisfying
(restricted) linear inequalities can be detected in $\EXPSPACE$. 
For this, they use a Rackoff-style induction to bound the length of such runs. 
We devise a similar Rackoff-style induction to work with SNLS instead of
ordinary linear inequalities. 
Different compared to the earlier works is the fact that our witnesses contain $\omega$-counters, which may change when invoking the induction hypothesis.
Moreover, we need to use a result of Demri on selective unboundedness~\cite[Theorem 4.6]{Demri13} (in $\EXPSPACE$ in the general case and $\PSPACE$ in fixed dimension) to check the coverability of our witnesses.

\newcommand{\aval}{\mathit{c}}

\newcommand{\ints}{\mathbb{Z}}
\newcommand{\rats}{\mathbb{Q}}

\newcommand{\arat}{\mathit{t}}
\newcommand{\anint}{a}
\newcommand{\anintp}{b}
\newcommand{\apoly}{p}
\newcommand{\apolyp}{q}
\newcommand{\avar}{x}
\newcommand{\avarp}{y}
\newcommand{\avarpp}{x}

\newcommand{\areal}{r}
\newcommand{\apolyset}{P}

\newcommand{\param}{\avarpp}

\newcommand{\apolyof}[1]{\apoly(#1)}
\newcommand{\apolypof}[1]{\apolyp(#1)}

\newcommand{\asol}{\ourbold{s}}

\newcommand{\degof}[1]{\deg(#1)}

\newcommand{\maxcof}[1]{\maxc(#1)}

\newcommand{\aset}{\mathit{S}}

\newcommand{\binencof}[1]{\mathit{bin}(#1)}

\newcommand{\adjof}[1]{\adj(#1)}
\newcommand{\dimof}[1]{\dim(#1)}
\newcommand{\detof}[1]{\det(#1)}
\newcommand{\sgnof}[1]{\sgn(#1)}
\newcommand{\rowof}[1]{\mathit{row}(#1)}
\newcommand{\colof}[1]{\mathit{col}(#1)}

\newcommand{\normparof}[2]{|\!|#2|\!|_{#1}}
\newcommand{\normof}[1]{\normparof{}{#1}}
\newcommand{\binarysizeof}[1]{\normparof{2}{#1}}
\newcommand{\unarysizeof}[1]{\normparof{1}{#1}}

\section{Preliminaries}

\subparagraph{B\"uchi VASS}
A \emph{B\"uchi vector addition system with states (B\"uchi VASS)}
of dimension~$d\in\mathbb{N}$ over an alphabet $\Sigma$
is a tuple $\avass=(Q, q_0, \Sigma, T, F)$.
It consists of a finite set of states $Q$,
an initial state $q_0\in Q$,
a set of final states $F\subseteq Q$,
and a finite set of transitions $T \subseteq\ Q\times \Sigma^*\times \mathbb{Z}^d\times Q$.
The size of the B\"uchi VASS is $|\avass|:=|Q|+|F|+\sum_{(q, w,\delta, q')\in T}\big(|w|+\binarysizeof{\delta}\big)$.
By $\binarysizeof{\delta}$, we mean the size of the binary encoding of $\delta$.
Since we only consider B\"uchi VASS in this paper, we often simply call them VASS.
If $d=0$, we call $\avass$ a \emph{B\"uchi automaton}.

The semantics of the B\"uchi VASS is defined over its \emph{configurations}, which are elements of $Q\times\mathbb{N}^d$.
The \emph{initial configuration} of $\avass$ is $(q_0, \bzero)$.
We lift the transitions of the B\"uchi VASS to a relation over configurations
$\rightarrow\ \subseteq\ Q\times\mathbb{N}^d\times\Sigma^*\times Q\times\mathbb{N}^d$ as follows: 
$(q, \bmm)\xrightarrow{w}(q', \bmm')$ if there is $(q, w, \delta, q')\in T$ such that $\bmm'=\bmm+\delta$.
A \emph{run} of the B\"uchi VASS is a (possibly infinite) sequence of configurations of the form
$\sigma=(p_0, \bmm_{0})\xrightarrow{w_1}(p_1, \bmm_1)\xrightarrow{w_2}\cdots$.

A run $\sigma$ is \emph{accepting} if it starts from the initial configuration
and visits final states infinitely often,
meaning there are infinitely many configurations $(q, \bmm)$ in $\sigma$ with $q\in F$. 
The run is said to be \emph{labeled} by the word $w=w_0w_1\cdots$ in $\Sigma^{\omega}$. 
The \emph{language} $L(\cV)$ of the B\"uchi VASS consists of all infinite
words that label an accepting run.

An infinite-word language $L \subseteq \Sigma^\omega$ is called \emph{regular}
if it is accepted by a B\"uchi automaton.
As we only consider infinite-word languages, we just call them languages.

\subparagraph{Arithmetic} Our approach to regular separability in B\"{u}chi VASS rests on a result about solutions to singly non-linear systems of inequalities. This also requires some terminology.

We define the integers, rationals, polynomials, and matrices together with the operations we need to perform on them.
Let $\anint\in\ints$ be an integer. 
Its size $\binarysizeof{\anint}=\sizeof{\binencof{\anint}}$ is the length of its binary encoding. 
We also use $\unarysizeof{\anint}$ to denote the size of the unary encoding. 
This is the absolute value plus an extra bit for the sign. 
A polynomial with integer coefficients $\apoly\in\ints[\avar]$ is a sum $\sum_{i=0}^{k}\anint_i \avar^i$ with $\anint_0,\ldots, \anint_k\in\ints$ and $\anint_k\neq 0$ if $k>0$. 
We define $\unarysizeof{\apoly}=\sum_{i=0}^k\unarysizeof{\anint_i}$ and similar for $\binarysizeof{\apoly}$. 
The \emph{degree} of the polynomial is $\degof{\apoly}=k$, its \emph{maximal coefficient} is $\maxcof{\apoly}=\max_{i\in[0, k]}\unarysizeof{a_i}$.  
Note that $\unarysizeof{\apoly} \leq (\degof{\apoly} + 1) \cdot \maxcof{\apoly}$. 
A real number $r\in\R$ with $\apoly(r)=0$ is called a \emph{root} of the polynomial. 
Let $\aset$ be a set with a size function $\normof{-}$ defined on it.
We consider matrices $\bA\in\aset^{m\times n}$ over~$\aset$, 
and define their size $\normof{\bA}$ by summing up the sizes of the entries.
We use $\rowof{\bA}=m$ and $\colof{\bA}=n$. 
When $\aset = \ints[\avar]$, we also use $\degof{\bA}$  for the highest degree of a polynomial in $\bA$ and $\maxcof{\bA}$ for the maximal coefficient of a polynomial in $\bA$.
Pairs $(s_1, s_2)\in\aset\times\aset$ form a special case with size $\normof{(s_1, s_2)}=\normof{s_1}+\normof{s_2}$. 
In particular, a rational number $\arat\in\rats$ is a pair $\arat = \frac{\anint}{\anintp}$ of integers $\anint, \anintp\in\ints$ with $\unarysizeof{\arat}=\unarysizeof{\anint}+\unarysizeof{\anintp}$, and similar for $\binarysizeof{\arat}$.

We perform addition $\anint+\anintp$ and multiplication $\anint\cdot \anintp$ among integers, rationals, polynomials, and matrices. 
These operations can be executed in time polynomial in $\binarysizeof{\anint}+\binarysizeof{\anintp}$. 
The same holds for the comparison $\anint\geq \anintp$ among integers and rationals.
We also add, multiply, and compare integers and rationals with $-\infty$ and $\infty$.
The definitions are as expected.

\section{Main results}\label{main-results}
A language $\aregsep\subseteq\analph^{\omega}$ is said to \emph{separate} languages $\alang_1, \alang_2\subseteq\analph^{\omega}$, if $\alang_1\subseteq \aregsep$ and $\aregsep\cap \alang_2=\emptyset$. 
We call $\alang_1$ and $\alang_2$ \emph{regular separable}, denoted by $\alang_1\separable\alang_2$, if there is a separator $\aregsep$ that is a regular language. 
The problem we address is the \emph{regular separability problem} for B\"{u}chi VASS:
\begin{description}
\item[Given] Two VASS $\avass_1$ and $\avass_2$ over some alphabet $\Sigma$.
\item[Question] Does $\langof{\avass_1}\separable\langof{\avass_2}$ hold?
\end{description}
We also consider the variants of this problem where the inputs are of fixed dimension:
For a fixed number $d \in \nat\setminus\set{0}$, the \emph{$d$-dimensional regular separability problem} for B\"{u}chi VASS is the same problem as above, except that the input VASS $\avass_1$ and $\avass_2$ are restricted to be of dimension at most $d$.
Our first main result is the following:
\begin{theorem}\label{Theorem:UpperBound}
The regular separability problem for B\"{u}chi VASS is $\EXPSPACE$-complete. Moreover, the $d$-dimensional regular separability problem is $\PSPACE$-complete for all $d\geq1$.
\end{theorem}

As mentioned above, the proof is based on a small model property for what we call singly non-linear systems of inequalities. We expect this result to be of independent interest.
Formally, a \emph{singly non-linear system (SNLS)} is a system of inequalities of the form
\begin{align*}
\bA(\param)\cdot \by \geq \bb(\param) \wedge \by \geq \zerovec\ .
\end{align*}
Here, $\bA(\param) \in \Z[\param]^{m \times n}$ is an $m\times n$ matrix over the set of polynomials with integer coefficients in variable $\param$,
and $\bb \in \Z[\param]^m$ is a vector of polynomials. 
We also write an SNLS as $\cS = (\bA(\param), \bb(\param))$, or $\cS(\param, \by)$ to emphasize the variables. 
A \emph{solution} to $\cS$ is a pair $(\arat, \asol)\in\Q\times \Q^n$ that satisfies $\bA(\arat)\cdot\asol\geq \bb(\arat)\wedge \asol\geq \zerovec$.  
If a solution exists, we call the system \emph{feasible}.

Our second main result is a bound on the size of least solutions.
\begin{theorem}\label{solution-size}
	If the SNLS $\cS$ is feasible, then it has a solution $(\arat, \asol)$, where all components of $\asol$ have the same denominator, with $\unarysizeof{\arat}, \unarysizeof{\asol}\in (\colof{\cS}\cdot \degof{\cS} \cdot \maxcof{\cS})^{\bigoof{\degof{\cS}^2 \cdot \rowof{\cS}^4}}$.
\end{theorem}
\Cref{solution-size} implies that a feasible system $\cS$ always has a solution of size at most singly exponential in $\unarysizeof{\cS}$.
This gives an upper bound on the complexity of feasibility.
\begin{corollary}
Feasibility of SNLS is in $\NP$. 
\end{corollary}
The reader may have noted that SNLS are more general than the non-linear systems we are confronted with when checking separability. 
There are at least two arguments in support of the generalization.  
First, non-linearity is not well-understood, and we believe a class of systems that admits an efficient algorithm for checking feasibility will find its applications.
Second, the generalization only adds little complexity to the proof or, phrased differently, the special case already needs most considerations. 

\subparagraph{Organization} The remainder of the paper is organized as follows. In \cref{section:quant-elim}, we prove \cref{solution-size} and in \cref{section:search}, we show \cref{Theorem:UpperBound}.

\section{Singly Non-Linear Systems}\label{section:quant-elim}

In this section, we prove \cref{solution-size}.

\subparagraph{Some notation}
By $\bA(\arat)$ or $\mathit{eval}(\bA(\param), \arat)$ we mean the matrix with rational entries that results from $\bA(\param)$ by evaluating all polynomials at~$\arat$. 
Let $\bA\in R^{n\times n}$ be a square matrix over some ring $R$. In our exposition, we will consider matrices over the rings $\Z$, $\Q$, and $\Z[x]$.
We write $\detof{\bA}$ for the determinant, and recall that if $R$ is a field (such as $\Q$), then $\bA$ is invertible if and only if $\detof{\bA}\neq 0$.
The \emph{adjugate} (also called \emph{classical adjoint}) of $\bA$ is the matrix $\adjof{\bA}\in R^{n\times n}$ with $\adjof{\bA}[j,i]=(-1)^{i+j}\det(\bA_{ij})$, where $\bA_{ij}$ is the matrix obtained from $\bA$ by removing the $i$-th row and the $j$-th column.  It is well-known that then $\bA\cdot\adj(\bA)=\det(\bA)\cdot \bI$, where $\bI$ is the identity matrix in dimension $n$. In particular, if $\bA$ is invertible, its inverse can be computed as $\bA^{-1}=\frac{\adjof{\bA}}{\detof{\bA}}$~\cite[Chapter XIII, Prop. 4.16]{Lang2002}. 

In upper bound arguments, we will use the well-known Leibniz formula for
determinants, which says $\detof{\bA}=\sum_{\sigma\in S_n}\sgnof{\sigma}\cdot
\prod_{i=1}^n \coordacc{\bA}{i, \sigma(i)}$~\cite[Chapter XIII, Prop.
4.6]{Lang2002}.  Here, $S_n$ is the set of all permutations of $[1,n]$ and $\sgnof{\sigma}\in\{-1,1\}$ is the sign of $\sigma\in S_n$.    

\subparagraph{Bounding solutions}
For the proof of \cref{solution-size}, we proceed in two steps.
We first show that if an SNLS $\cS(\param, \by)$ is feasible, then we find a small rational $\arat$ for $\param$ such that the system $\cS(\arat, \by)$ is feasible.
This system is the result of evaluating all polynomials in $\cS$ at $\arat$, and thus having only $\by$ as the variables.
\begin{lemma}\label{exist-solutions-size}
  If the SNLS $\cS(\param, \by)$ is feasible, then there is a number $\arat\in\rats$ with $\unarysizeof{\arat}\in (\colof{\cS}\cdot \degof{\cS} \cdot \maxcof{\cS})^{\bigoof{\degof{\cS} \cdot \rowof{\cS}^3}}$
  such that $\cS(\arat, \by)$ is feasible.
\end{lemma}
Lemma~\ref{exist-solutions-size} is non-trivial and will occupy almost this entire section. 
To explain our approach, note that the feasibility of $\cS(\param, \by)$ is equivalent to the feasibility of $\exists \by.\cS(\param, \by)$. 
Our first step is to eliminate the quantifier and determine a new formula $\qelimformdef$ in which $\by$ no longer occurs and that is equivalent to the previous one over the rationals, $\exists \by.\cS(\param, \by)\ratequiv\qelimformdef$.
The equivalence says that for every $\arat\in\rats$, we have $\arat\models \exists \by.\cS(\param, \by)$ if and only if $\arat\models \qelimformdef$.

The second step for \cref{exist-solutions-size} is to show that if the new formula holds, then we find a small solution for $t$. To this end, we will apply the Root Separation Theorem, which provides a lower bound on the distance between distinct real roots of polynomials. After establishing \cref{exist-solutions-size}, we obtain \cref{solution-size} (at the end of this section) by taking the $t$ provided by \cref{exist-solutions-size}, and pair it with the $\bs\in\Q^n$, which must exist according to the quantifier elimination done in the first step of \cref{exist-solutions-size}.
\subsection{Quantifier Elimination}
We show how to remove the quantifier from $\exists \by.\cS(\param, \by)$ with a tailor-made quantifier elimination algorithm. 
The fact that %
quantifier elimination is possible in this setting came as a surprise to us, given the non-linear nature and the setting of rationals. For example, the real closed field $(\R;+,\cdot,<,0,1)$ admits quantifier elimination by a classical result of Tarski~\cite[Theorem 3.3.15]{Marker2002}, but this is not true for the ordered field $(\Q;+,\cdot,<,0,1)$ of rationals~\cite[Theorem 2]{macintyre1983elimination} (see~\cite[p. 71--72]{Marker2002} for a simple example). This means, there are first-order formulas over $(\Q;+,\cdot,<,0,1)$ that have no quantifier-free equivalent. However, we show that if we existentially quantify the linear variables in the formulas induced by SNLS, then those quantifiers can be eliminated.

The precise formulation of the result needs some notation.
A \emph{lower bound constraint} has the form $\apoly(\param)\geq 0$ or $\apoly(\param) > 0$ with $\apoly\in\ints[\param]$ a polynomial with integer coefficients.   
The formula $\qelimformdef$ that we want to obtain takes the form
$\bigvee_{i\in I} \bigwedge_{j\in J_i} \qelimform_{i, j}(\param)$, where the formulae $\qelimform_{i, j}(\param)$ are lower bound constraints.
We call it a DNFLB, short for \emph{disjunctive normal form with lower bound constraints as the literals}.  
We may also omit $\param$ and write $\qelimform$. 
We use $\degof{\qelimform}$ and $\maxcof{\qelimform}$ for the maximal degree resp.\ coefficient of a polynomial in~$\qelimform$.  
\begin{theorem}\label{Theorem:QElim}
For every SNLS $\cS(\param, \by)$, there is a DNFLB~$\qelimformdef$ with $\exists \by.\cS(\param, \by)\ratequiv \qelimformdef$, $\degof{\qelimform}\in\bigoof{\rowof{\cS} \cdot \degof{\cS}}$, and $\maxcof{\qelimform}\in (\colof{\cS}\cdot\degof{\cS} \cdot \maxcof{\cS})^{\bigoof{\rowof{\cS}^2}}$.
\end{theorem}  
Since our intention is to bound the solutions $\arat$ to variable $\param$, the given estimations on the degree and the maximal coefficient suffice for us. 
The proof actually gives an algorithm to compute $\qelimform$ which runs in time exponential in the dimension of $\cS$, but we do not need the effectiveness here.
In the proof of \cref{Theorem:QElim}, we will use a standard fact about polyhedra:
\begin{lemma}\label{minimal-face}
	Suppose $\bD\in\Q^{m\times n}$ and $\bc\in \Q^m$. If the system
	$\bD\cdot\bx\ge\bc$ has a solution in $\Q^n$, then there is a solution
	$\bs\in\Q^n$ that also satisfies $\bD'\cdot\bs=\bc'$, where $(\bD',\bc')$ is
	a subset of the rows of $(\bD,\bc)$ such that $\rank(\bD')=\rank(\bD)$.
\end{lemma}
\begin{proof}
	By well-known decomposition theorems about polyhedra, a polyhedron
	$P=\{\bs\in\Q^n\mid \bD\cdot\bs\ge\bc\}$ is non-empty if and only if it has
	a non-empty minimal face~\cite[Theorem 8.5]{Schrijver1986}. Moreover,
	minimal faces can be characterized as exactly the sets of the form
	$\{\bs\in\Q^n \mid \bD'\cdot\bs=\bc'\}$, where $(\bD',\bc')$ is a subset of
	the rows of $(\bD,\bc)$ such that $\bD'$ has the same rank as
	$\bD$~\cite[Theorem 8.4]{Schrijver1986}.
\end{proof}

We are ready to prove \cref{Theorem:QElim}:
\begin{proof}[Proof of \cref{Theorem:QElim}]
Let $\cS(\param, \by)=\bA(\param)\cdot \by \geq \bb(\param) \wedge \by \geq \zerovec$. 
To fix the dimension, let $\bA\in\Z[\param]^{m\times n}$. 
We can equivalently write the SNLS as $\cS'(\param, \by)=\bD(\param)\cdot\by\geq \bc(\param)$ with 
\begin{equation*}
\bD(\param)\ =\ \left(\begin{aligned}&\bA(\param)\\&\bI_n\end{aligned}\right)\in \Z[\param]^{(m + n) \times n}\qquad \bc(\param)\ =\ \left(\begin{aligned}&\bb(\param)\\&\bzero_n\end{aligned}\right)\in \Z[\param]^{m + n}\ ,
\end{equation*}
i.e.\ we glue the $n \times n$ identity matrix $\bI_n$ to the bottom of $\bA(\param)$ and extend $\bb(\param)$ by $n$ zeros.

Assume $\cS'$ is feasible and the solution for $\param$ is $\arat\in\rats$.
	By \cref{minimal-face} and since $\bD(\arat)$ has rank $n$, we can select a subset of $n$ rows of $\bD(\arat)$ and of $\bc(\arat)$ such that the smaller system has a solution, even with equality. More formally, for any subset $R\subseteq[1,m+n]$, denote by $\bD_R(\arat)$ (resp.\ $\bc_R(\arat)$) the matrix (resp.\ vector) obtained by selecting only the rows in $R$ from $\bD(\arat)$ (resp.\ $\bc(\arat)$). Then \cref{minimal-face} tells us that there is an $\asol\in\Q^n$ with $\bD_R(\arat)\cdot\asol=\bc_R(\arat)$, where $\bD_R(\arat)$ has rank $n$. In particular, $\bD_R(\arat)$ is invertible and thus $\det(\bD_R(\arat))\ne 0$. The fact that $\bD_R(t)$ is invertible allows us to write $\asol=\bD_R(\arat)^{-1}\cdot\bc_R(\arat)$, which will be key for our quantifier elimination.
The argumentation shows that for every $t\in\Q$, $\exists \by.\cS'(t,\by)$ is equivalent to the condition
\begin{equation}
	\bigvee_{\substack{R \subseteq [1,m+n]\\ \sizeof{R} = n}} \detof{\bD_R(t)}\neq 0\ \wedge\ \bD(t)\cdot \bD_R(t)^{-1}\cdot \bc_R(t)\geq \bc(t)\ . \label{condition-feasible}
\end{equation}
	Here, of course, we only know that $\bD_R(\arat)^{-1}$ exists when $\det(\bD_R(\arat))\ne 0$.
	To express \eqref{condition-feasible} using polynomials, we employ the identity $\bD_R(\arat)^{-1}=\tfrac{\adjof{\bD_R(\arat)}}{\detof{\bD_R(\arat)}}$ whenever $\bD_R(\arat)$ is invertible (equivalently, whenever $\detof{\bD_R(\arat)}\ne 0$). Thus, the set of all $\arat$ with $\exists\by\colon \cS'(\arat,\by)$ can be defined by the following DNFLB $\Phi$:
  \begin{align}
  \bigvee_{\substack{R \subseteq [1,m+n]\\ \sizeof{R} = n}}~
    &\Big(
      \detof{\bD_R(\param)} > 0
      ~\wedge~ \bD(\param) \cdot \adjof{\bD_R(\param)} \cdot \bc_R(\param) \geq \detof{\bD_R(\param)} \cdot \bc(\param)
    \Big) \label{phi-qelim}
 \\
    &\vee \Big(
      \detof{\bD_R(\param)} < 0 
      ~\wedge~ \bD(\param) \cdot \adjof{\bD_R(\param)} \cdot \bc_R(\param) \leq \detof{\bD_R(\param)} \cdot \bc(\param)
    \Big)\ ,\notag  \end{align}
where indeed all conditions are of the form $p(\param)\ge 0$ or $p(\param)>0$ for some polynomials~$p$. 
	Note that here, we distinguish the cases $\detof{\bD_R(\param)}<0$ and $\detof{\bD_R(\param)}>0$ because moving a negative $\detof{\bD_R(\param)}$ to the other side of the inequality changes $\ge$ to $\le$. Moreover, note that (in contrast to \eqref{condition-feasible}) in the formulation \eqref{phi-qelim}, all terms are well-defined, independenly of whether the current choice of $R$ makes $\bD_R(x)$ invertible or not.
To be precise, we obtain the DNFLB by subtracting the right-hand sides of the inequalities from the left-hand sides and multiplying the result by $-1$ to invert the inequality where necessary.
The above form will suffice to give an estimate on the maximal degree and the maximal coefficient. 

It is now clear that the coefficients (resp.\ degrees) appearing in $\Phi$ are exponential (resp.\ polynomial) in the bitsize of $\mathcal{S}$. The precise bounds promised in the Theorem are straightforward to deduce from standard bounds on determinants, see \conferenceFull{the full version}{\cref{appendix-snls}} for details.
\end{proof}
\subsection{Root Separation}
To show \cref{exist-solutions-size}, it remains to be shown that any feasible DNFLB $\Phi(\param)$ has a solution that is
exponentially bounded.  The key observation is that if $r$ and $r'$ are
adjacent roots of a polynomial $p(\param)\in\Z[\param]$ and a constraint $p(\param)\ge 0$ or
$p(\param)>0$ is satisfied for some $\arat$ for $\param$ with $r<\arat<r'$,
then any number $\arat'$ in the open interval $(r,r')$ will
also satisfy the constraint: The polynomial does not change its sign between
$r$ and $r'$. Thus, we can think of $\R$ as being split into (i)~roots of $p$
and (ii)~intervals between roots of $p$ (and the infinite intervals below the smallest and above the largest root).
Then whether $\arat \in \Q$ satisfies $p(\param)\ge 0$ or $p(\param)>0$ only depends on
which of those parts of $\R$ the number $\arat$ belongs to. This remains true if we refine this
decomposition of $\R$ according to \emph{all} polynomials occurring in $\Phi$.

In order to construct rational numbers with small numerator and denominator in intervals $(r,r')$, we will rely on a Root Separation Theorem, saying that polynomial roots are not too close. More specifically, we use Rump's Bound~{\cite[Theorem 8.5.5]{Mishra93}}:
\begin{restatable}[Rump's Bound~{\cite[Theorem 8.5.5]{Mishra93}}]{theorem}{rump}\label{rump}
  Suppose $r,r'\in\R$ are distinct roots of a polynomial $\apoly(\param) \in
\Z[\param]$ with degree~$d\in\nat$.  Then $|r-r'| > (d^{d+1}(1+
\unarysizeof{\apoly(\param)})^{2d})^{-1}$. 
\end{restatable}

We will also use an elementary fact about rational roots of integral polynomials. It is known as the Rational Root Theorem or Integral Root Test~\cite[Chapter IV, Prop. 3.3]{Lang2002}:
\begin{lemma}[Rational Root Theorem~{\cite[Chapter IV, Prop. 3.3]{Lang2002}}]\label{rational-root-theorem}
	Let $p(x)=c_nx^n+\cdots+c_0\in\Z[x]$ be a polynomial. If $r=a/b$ is a root of $p$ with $a,b$ co-prime, then $a$ divides $c_0$ and $b$ divides $c_n$. In particular, $|a|,|b|\le \maxc(p)$.
\end{lemma}

Finally, we need a standard bound on all real roots of a polynomial~\cite[Corollary 8.3.2]{Mishra93}. This is known as Cauchy's bound.
\begin{lemma}[Cauchy's Bound~{\cite[Corollary 8.3.2]{Mishra93}}]\label{cauchy-bound}
If $r\in\R$ is a root of a polynomial $p\in\Z[\param]$, then $|r|\le 1+\unarysizeof{p}$.
\end{lemma}

Let $r_1<\cdots<r_k\in\R$ be all the real roots of polynomials occurring in $\Phi$.
Observe that if $\arat\in\Q$ satisfies $\Phi(\param)$ and $\arat\in(r_i,r_{i+1})$, then any
rational number in $(r_i,r_{i+1})$ must satisfy $\Phi(\param)$, because none of the
polynomials in $\Phi$ changes its sign between $r_i$ and $r_{i+1}$. This allows
us to bound a rational solution, by distinguishing the following cases:

\begin{enumerate}
\item Suppose $\Phi(\param)$ is satisfied by some rational root $r_i$ of $p$ in $\Phi$. Write $r_i=\tfrac{a}{b}$ with $a,b$ co-prime. Then the Rational Root Theorem (\cref{rational-root-theorem}) implies $|a|,|b|\le \maxc(p)$.
\item Suppose $\Phi(\param)$ has a rational solution in some interval $(r_i,r_{i+1})$. Since $r_i$ and $r_{i+1}$ are the roots of some polynomials $p,q$ in $\Phi(\param)$, as observed above, any rational number in $(r_i,r_{i+1})$ is also a solution to $\Phi(\param)$. Note that $r_i,r_{i+1}$ are roots of $p(\param)\cdot q(\param)$ and thus by \cref{rump}, we have $|r_i-r_{i+1}|>\tfrac{1}{b}$ for some $b\in\Z$ that is exponentially bounded. Thus, there is an integer $a\in (br_i,br_{i+1})$. Note that then $\tfrac{a}{b}$ belongs to the interval $(r_i,r_{i+1})$ and thus satisfies $\Phi(\param)$. Moreover, by the Cauchy Bound (\cref{cauchy-bound}), we also have an exponential bound $U\in\R$ on $|r_i|,|r_{i+1}|$ and thus on $|a|\le |b|U$.  
\item Suppose $\Phi(\param)$ has a rational solution $\arat$ outside of $[r_1,r_k]$. If $\arat > r_k$ then every rational number in $[r_k,\infty)$ is also a solution. Moreover, by \cref{cauchy-bound}, any rational number $\arat'$ with $\arat'>1+\unarysizeof{p}$ for every polynomial $p$ occurring in $\Phi$ can be chosen, e.g.\ $\arat' = 2+c$, where $c = \max\{\unarysizeof{p} \mid p$ polynomial in $\Phi \}$. On the other hand, if $\arat < r_1$, then $\arat' = - (2+c)$ is a solution by an analogous argument.
\end{enumerate}

This proves that any feasible $\Phi(\param)$ has a rational solution that is exponentially bounded, which is what we will use in our application to B\"{u}chi VASS. The precise bounds of \cref{exist-solutions-size} are shown in \conferenceFull{the full version}{\cref{appendix-exist-solutions-size}}.

\subparagraph{Proof sketch for \cref{solution-size}} For showing
\cref{solution-size}, we can now use the fact that if $t\in\Q$ admits a
solution $(t,\bs)$, then by our argument in the proof of \cref{Theorem:QElim},
$\bs^*:=\tfrac{\adj(D_R(t))}{\det(D_R(t))}\cdot \bc_R(t)$ is also a
solution, for some subset $R\subseteq[1,m+n]$. This means that we can apply the
bound on $t$ and the bounds on $\adj(\bD_R(x))$ and $\det(\bD_R(x))$ established in
the proof of \cref{Theorem:QElim} to bound the solution $\bs^*$.
We can ensure that all components of $\bc_{R}(t)$ have the same denominator by increasing the bit size at most $\degof{\cS}$-fold.
If we compute $\bs^{*}$ starting from such a vector, we get an $\bs^{*}$ where all components have the same denominator.
Since the
entries in $\bD_R$ all appear in $\cS$ and it is well-known that the
determinant has polynomial bit size in the bit size of a matrix, it follows
that there exists a solution $(t,\bs)$ of polynomial bit size.
The precise
bounds promised in \cref{solution-size} are derived in
\conferenceFull{the full version}{\cref{appendix-solution-size}}.

\newcommand{\ekmark}[1]{\color{purple}#1\color{black}~}
\newcommand{\loopvector}{\mathbf{o}}
\newcommand{\valuevec}{\bz}
\newcommand{\countersp}{J}
\newcommand{\reachof}[2]{\mathsf{reach}(#1.#2)}

\section{$\omega$-Regular Separability}\label{section:search}

We use the results from \cref{section:quant-elim} to prove \cref{Theorem:UpperBound}.
Note that the lower bounds in \cref{Theorem:UpperBound} easily follow from \cite{Baumann23}. First, that paper already shows $\PSPACE$-completeness of regular separability for one-dimensional B\"{u}chi VASS, which yields the $\PSPACE$ lower bound for fixed dimension $\ge 1$. In fact, their argument also yields $\EXPSPACE$-hardness in the general case: The full version~\cite[Appendix E.1]{Baumann23full} describes a simple reduction from intersection emptiness of one-dimensional VASS that accept by final state to regular separability of B\"{u}chi VASS, and the construction is the same in higher dimension.
This yields $\EXPSPACE$-hardness of regular separability of B\"{u}chi VASS, since intersection emptiness of VASS of arbitrary dimension that accept by final state is $\EXPSPACE$-hard \cite{Lipton76}.

It remains to prove the upper bounds in \cref{Theorem:UpperBound}.
To adequately formulate our proofs, we need to introduce additional VASS-related concepts.

\subparagraph{More on B\"uchi VASS}
Let $\avass=(Q, q_0, \Sigma, T, F)$ be a B\"uchi VASS.
Consider a (possibly infinite) run 
$\sigma=(p_0, \bmm_{0})\xrightarrow{w_1}(p_1, \bmm_1)\xrightarrow{w_2}\cdots$ of $\avass$.
The sequence of transitions in $\sigma$ is called a \emph{path} and has the form
$\apath=(p_0, w_1, \delta_1, p_1)(p_1, w_2, \delta_1, p_2)\ldots$.
If a path is finite and the source state of its first transition coincides with the target state of its last transition, then we call it a \emph{loop}. %
Since a run is uniquely determined by the start configuration and its sequence of transitions, we also denote a run by $\sigma=(p_0, \bmm_{0}).\apath$. 
If $\sigma$ is finite and $(p_\ell, \bmm_\ell)$ is its last configuration, then we sometimes write
$\sigma=(p_0, \bmm_{0}).\apath.(p_\ell, \bmm_\ell)$ to emphasize this.
The \emph{effect} $\vassbalanceof{\apath}$ of some finite path $\apath=(p_0, w_1, \delta_1, p_1)\ldots(p_{\ell-1}, w_\ell, \delta_\ell, p_\ell)$ is the sum of all induced counter changes, formally $\vassbalanceof{\apath} = \sum_{1\leq i\leq \ell}\delta_{i}$. 

Recall that configurations of the B\"uchi VASS $\avass$ are elements of the set $Q\times\mathbb{N}^d$.
We call the second component in a configuration the \emph{counter valuation}
and refer to the $i$-th entry as the \emph{value of counter $i$}.
For a configuration $\aconfig$ and a set of counters $I \subseteq [1,d]$ we also use $\aconfig[I]$ to denote the counter valuation of $\aconfig$ restricted to the counters in $I$.
A configuration $(q,\bmm)$ is \emph{coverable} in $\avass$ if there is a run starting in the initial configuration $(q_0,\bzero)$ and reaching a configuration $(q,\bmm')$ with $\bmm' \geq \bmm$.
Here, $\geq$ is defined component-wise.

Moreover we also consider a set of \emph{extended configurations} $\states\times\natomega^{d}$,
where $\natomega = \N \cup \{\omega\}$.
Here $\omega$ is used to represent a counter value that has become unbounded.
For an extended configuration $(q,\bmm)$ we use $\omega(q,\bmm) \subseteq [1,d]$ to denote the
set of counters valued $\omega$ in $\bmm$.
Comparisons and arithmetic operations between integer values and $\omega$ behave as expected,
treating $\omega$ as $\infty$.
Formally, $\omega \geq \omega$, $\omega \geq z$, and $\omega + z = \omega$ for all $z \in \Z$.
The \emph{size} of an extended configuration is $|(q,\bmm)| = \log_2|Q| + \binarysizeof{\bmm} + d$, where the extra bit per counter encodes whether it has value $\omega$ or not.
We also use the size of a unary encoding $\unarysizeof{(q,\bmm)} = |Q| + \unarysizeof{\bmm} + d$.

The transition relation is also lifted to extended configurations in the expected manner.
Formally, for $(q, \bmm)$, $(q',\bmm') \in Q \times \natomega^d$ we have
$(q, \bmm)\xrightarrow{w}(q', \bmm')$
if there is a transition $(q, w, \delta, q')\in T$ such that $\bmm'=\bmm+\delta$,
where addition between elements of $\natomega^d$ and $\Z^d$ is defined component-wise.
Furthermore, our definitions of runs, paths, loops, etc.\ carry over to \emph{extended} versions over
the set of extended configurations in a straightforward way.
More precisely, an \emph{extended run} is a sequence of extended configurations $\aconfig_1 \xrightarrow{w_1} \aconfig_2 \xrightarrow{w_2} \cdots$, an \emph{extended path} is the underlying sequence of transitions of an extended run, and an \emph{extended loop} is a finite extended path starting and ending in the same state.
To cover an extended configuration, intuitively, the $\omega$-counters need to become unbounded, and the remaining counters need to be covered.
Formally, an extended configuration $(q,\bmm)$ is \emph{coverable} in $\avass$ if for every $k \in \N$ there is a run starting in the initial configuration $(q_0,\bzero)$ and reaching a configuration $(q,\bmm_k) \in Q\times\N^d$ such that $\bmm_k[j] \geq k$ for every counter $j \in \omega(q,\bmm)$ and $\bmm_k[i] \geq \bmm[i]$ for every counter $i \in [1,d] \setminus \omega(q,\bmm)$.

Finally, we sometimes want to restrict only some counters of the VASS to stay non-negative.
In this case, we consider extended configurations in $\states\times\Z_{\omega}^{d}$,
where $\Z_{\omega} = \Z \cup \{\omega\}$.
We say an extended run $\sigma = \aconfig.\apath$ \emph{remains non-negative} on counters
$I \subseteq [1,d]$ if $\aconfig'[I] \subseteq \N^{|I|}$ for all extended configurations
$\aconfig'$ on $\sigma$.

\subparagraph{Dyck Language}
Towards the $\EXPSPACE$ upper bound of \cref{Theorem:UpperBound}, a first step is to reduce
the separability problem to a variant where one language is fixed to the Dyck language.
The  Dyck language $\dycklang{n}$ with $n$-letters is defined over the alphabet $\analph_n=\setcond{\alet_i, \abarlet_i}{i\in[1, n]}$. 
It contains those words $\aword$ where, for every prefix $\awordp$ with $\aword=\awordp.\awordpp$, we have at least as many letters $\alet_i$ as $\abarlet_i$. 
Thus, the letters behave like VASS counters and, indeed, the Dyck language is accepted by a single-state VASS $\dyckvass{n}$ with $n$ counters that increments the $i$-th counter upon seeing letter $\alet_i$ and decrements the $i$-th counter upon seeing $\abarlet_i$.
If a VASS is defined over the Dyck alphabet $\Sigma_n$, we also call it $n$-visible.
We will sometimes treat an $n$-visible VASS of dimension $d$ as a $(d+n)$-dimensional VASS, and refer to the additional $n$ counters as external.
Note that this amounts to forming the product with $\dyckvass{n}$.
Given a path $\apath$, we use $\exteffof{\apath}$ for the effect on the external counters in this product construction.
Moreover, we write $\allbalanceof{\apath}$ to denote the combined effect on both internal and external counters, i.e.\ the $(d+n)$-dimensional vector $(\vassbalanceof{\apath},\balanceof{\apath})$.

To avoid an exponential blow-up, our reduction uses a variant of VASS whose transitions are labeled by compressed words. %
Essentially, the reduction takes $\avass_1$ and $\avass_2$ and produces a VASS $\avass$ that is a product of $\avass_1$ and $\avass_2$. Moreover, it acts on its %
counters like $\avass_1$; the input labels
of $\avass$ correspond to the counter updates of $\avass_2$. Since the latter are binary-encoded, the new VASS will have binary encoded input words. Let us make this precise. A \emph{label-compressed VASS} (lcVASS) $\avass$ is a VASS, where the transitions are of the form $(p,a^m,\delta,q)\in Q\times\Sigma^*\times\Z^d\times Q$, where $a\in\Sigma$ and $m\in\N$ is given in binary. Thus, for an lcVASS $\avass$, we define its \emph{size} as $\sizeof{\avass} = |Q|+|F|+\sum_{(q,a^m,\delta,q')\in T} (\log_2(m)+\binarysizeof{\delta})$.
The reduction that fixes the Dyck language is captured by the following lemma.
\begin{lemma}[{\hspace{1sp}\cite[Lemma 3.4]{Baumann23}}]\label{Lemma:Transducers}
Given $\avass_1$ and $\avass_2$ over $\analph$, we can compute in time polynomial in $\sizeof{\avass_1}+\sizeof{\avass_2}$ an $n$-visible lcVASS $\avass$ %
so that $\langof{\avass_1}\separable\langof{\avass_2}$ if and only if $\langof{\avass}\separable\dycklang{n}$. Here, $n$ is the dimension of $\avass_2$. 
\end{lemma}
The polynomial time bound is not mentioned in~\cite{Baumann23}, but the simple construction they use (from~\cite{DBLP:conf/lics/CzerwinskiZ20}) clearly implies this bound.
The latter separability problem $\langof{\avass}\separable\langof{\dycklang{n}}$ has been studied closely in~\cite{Baumann23} as well. 
They first show that $\avass$ can be transformed so as to make the language $\langof{\avass}$ pumpable. 
For the resulting VASS, they show that $\langof{\avass}\separable\langof{\dycklang{n}}$ holds if and only if the so called Karp-Miller graph of $\KM{\avass}$ does not contain an inseparability witness.  
Unfortunately, the transformation required for pumpability involves another Karp-Miller graph construction, and therefore does not fit into the space bound we aim for (said graph can be of Ackermannian size in the worst case).
Instead, we reformulate the witness.
\begin{definition}\label{Def:bloom}
  Let $\avass$ be an $n$-visible $d$-dimensional VASS.
  An \emph{inseparability bloom} for~$\avass$ is a tuple 
  $\abloom = (\finalstate, \counters, \alpha, \beta, \gamma)$ 
  with $\finalstate$ a final state, loops $\alpha$, $\beta$, $\gamma$ starting and ending in $\finalstate$, and a partition of the counters $\omegacounters\discup\counters = [1, d+n]$ so that 
  \begin{enumerate}[(i)]
    \item for all $\rho=\alpha, \beta, \gamma$, we have $\coordacc{\allbalanceof{\rho}}{\counters}\geq 0$, 
    \item $\vassbalanceof{\alpha}+\vassbalanceof{\beta}+\vassbalanceof{\gamma}\geq 0$
    \item $\balanceof{\alpha}+\balanceof{\beta}\geq 0$, 
    \item there is a $\arat\in\rat$ with $\balanceof{\alpha}+\balanceof{\beta}+\balanceof{\gamma}= \arat\cdot\balanceof{\alpha}$.
  \end{enumerate}
  The \emph{size} of the bloom is $\sizeof{\abloom}=\log_2\sizeof{\states}+\sizeof{\counters}+\sizeof{\alpha}+\sizeof{\beta}+\sizeof{\gamma}$.  
  
  A \emph{stem for $\abloom$} is an extended run $\aconfig.\sigma$ of the product VASS $\avass \times \dyckvass{n}$ with the following properties: It ends in an extended configuration $\aconfig' = (q_f,\bmm)$ for some $\bmm \in \nat_\omega^{d+n}$ with $\omegacounters\subseteq\omegaof{\aconfig'}$, and the counters in $\counters\setminus \omegaof{\aconfig}$ remain non-negative when executing $\sigma$ from $\aconfig$ resp.\ $\alpha$, $\beta$, and $\gamma$ from $\aconfig'$. 

  An \emph{inseparability flower} $\aflower=(\aconfig.\sigma, \abloom)$ consists of an inseparability bloom and a suitable stem. The size is $\sizeof{\aflower}=\unarysizeof{\aconfig}+\sizeof{\sigma}+\sizeof{\abloom}$. 
  The flower is \emph{coverable} if the extended configuration $\aconfig$ is coverable in the product VASS $\avass \times \dyckvass{n}$.
\end{definition}
The following can be derived from the results in \cite{Baumann23},
refer to \conferenceFull{the full version}{\cref{appendix:search}} for the details.
\begin{lemma}\label{Lemma:InsepCharacterization}
Let $\avass$ be $n$-visible. 
We have $\langof{\avass}\inseparable \dycklang{n}$ if and only if some inseparability flower is coverable. 
\end{lemma}
Our main result is the following. 
Note: the unary counter encoding strengthens the bound. 
\begin{theorem}\label{Theorem:Bound}
If an inseparability flower is coverable in an $n$-visible lcVASS $\avass$ of dimension $d$, then there is one of size at most $\sizeof{\aflower}=2^{\sizeof{\avass}^{\bigoof{(d+n)^2}}}$.
\end{theorem}
\subparagraph{Main algortihm}
We now have all ingredients to formulate the algorithm that proves the upper bounds in Theorem~\ref{Theorem:UpperBound}. 
We first describe the $\EXPSPACE$ upper bound. Given $\avass_1$ and $\avass_2$ whose languages we wish to separate, we first compute the lcVASS $\avass$ using Lemma~\ref{Lemma:Transducers}. 
This takes poly time. 
The task is to check $\langof{\avass}\separable\dycklang{n}$, where $n$ is the dimension of $\avass_2$. 
Using Lemma~\ref{Lemma:InsepCharacterization}, we have to find an inseparability flower for $\avass$ that is coverable. 
Theorem~\ref{Theorem:Bound} bounds the size of the flowers we have to consider. 
We thus use non-determinism to find a flower of bounded size followed by a somewhat involved coverability check. 
Savitch's theorem~\cite{Savitch70} turns the non-deterministic algorithm into a deterministic one. 

We detect a flower of bounded size as follows. 
We first guess the final state $\finalstate\in\finalstates$ and the partitioning of the counters $\counters\discup\omegacounters$. 
With this information, we can guess the stem.

Towards obtaining a suitable stem, we start by guessing an extended configuration~$\aconfig$, whose non-$\omega$-entries are at most doubly exponentially large. 
We can store such configurations in exponential space.
If $\omegacounters\subseteq\omegaof{\aconfig}$ fails, we abort. 
We now guess a path $\sigma$ of doubly exponential length from $\aconfig$ to a configuration $\aconfig'$. 
As we proceed, we store the length of the path, which only needs exponential space. 
We abort, if one of the following happens while guessing $\sigma$: a counter from $\counters\setminus\omegaof{\aconfig}$ becomes negative, $\sigma$ becomes too long, or the last state on $\sigma$ is different from $\finalstate$.
If we have not aborted until now, we have determined $\aconfig.\sigma.\aconfig'$ that may serve as a stem for a bloom with final state $\finalstate$ and partition $\counters\discup\omegacounters$. 

Given the stem, we can finish the construction of the bloom
by guessing the cycles $\rho=\alpha, \beta, \gamma$. 
The reason we proceed in this order is the following. 
The cycles are too long to be stored in exponential space. 
Instead, we check the non-negativity required by a stem on-the-fly, while constructing the cycles. 
To do so, we need the configuration~$\aconfig'$, which we can store upon finishing the guess of $\sigma$ above.
While guessing the cycles, we store their length and the numbers $\inteffof{\rho}$ and $\exteffof{\rho}$. 
It is readily checked that these numbers are bounded by
\begin{align*}
2^{\sizeof{\avass}}\cdot 2^{\sizeof{\avass}^{\bigoof{(d+n)^2}}}= 2^{\sizeof{\avass}^{\bigoof{(d+n)^2}}}\ .
\end{align*}
This means we can store them in exponential space.
We abort, if one of the following applies: an intermediate valuation becomes negative on $\counters\setminus\omegaof{\aconfig}$, the path becomes too long, or the last state is different from $\finalstate$. 
We compute the operations and comparisons required by (i) to (iv) in \cref{Def:bloom}.
For (iv), we start with counter $d+1$ and store the quotient $\frac{\coordacc{\balanceof{\alpha}}{d+1}+\coordacc{\balanceof{\beta}}{d+1}+\coordacc{\balanceof{\gamma}}{d+1}}{\coordacc{\balanceof{\alpha}}{d+1}}$ as the rational number $\arat$. 
As the numerator and denominator will be at most doubly exponential, we can store the number in exponential space. 
For the remaining entries, we only perform the required comparison. 
As noted above, they can be executed in polynomial time.
If a comparison fails, we reject. 
If we have not rejected up to now, we have determined an inseparability flower while using space $\sizeof{\avass}^{\bigoof{(d+n)^2}} \leq 2^{\sizeof{\avass} \cdot \bigoof{(d+n)^2}}=2^{\mathit{poly}(\sizeof{\avass_1}+\sizeof{\avass_2})}$. 

It remains to check whether this flower is coverable. 
There are two challenges.
First, since $\aconfig$ is extended, we have a combination of a simultaneous unboundedness and a coverability problem.  
Second, the non-$\omega$-entries in $\aconfig$ may be doubly exponentially large.  
We reduce the problem to simultaneous unboundedness, which is in $\EXPSPACE$ as shown by Demri~\cite[Theorem 4.6(I)]{Demri13}.  
The reduction uses a simple gadget that subtracts the counter valuation to be covered and, if succcessful, makes a new target counter $j$ unbounded.
In the end we check simultaneous unboundedness of $\omega(\aconfig) \cup \{j\}$.
To handle the large values, we utilize Lipton's construction~\cite{Lipton76},
which allows a VASS to simulate $\EXPSPACE$-computations.
As a remark, simultaneous unboundedness cannot be expressed by a polynomial-sized formula in Yen's logic~\cite{Yen92}, a fact that was first observed in \cite{Demri13}, which means we cannot just invoke the bounds in~\cite{AtigH11}. 

For the $\PSPACE$ upper bound, we merely need to observe that in fixed dimension, all our bounds on counter values become singly exponential. Moreover, for the gadget that subtracts counter values, we do not need the Lipton construction, as we can subtract exponentially bounded values directly using transitions. Finally, in fixed dimension, the simultaneous unboundedness check is also possible in $\PSPACE$, as shown by Demri~\cite[Theorem 4.6(II)]{Demri13}.

The remainder of the paper proves Theorem~\ref{Theorem:Bound}. 
SNLS help us deal with Constraint~(iv).

\newcommand{\abound}{\mathit{b}}
\newcommand{\largebound}{\mathit{u}_{\mathit{num}}}
\newcommand{\largeboundp}{\mathit{v}}
\newcommand{\distinguished}{\mathit{D}}
\newcommand{\dloopboundof}[2]{\mathit{u}_{\mathit{len}}}
\newcommand{\dloopbounddef}{\dloopboundof{\abound}{\rackcounter}}
\newcommand{\effectof}[1]{\allbalanceof{#1}}
\newcommand{\parikhof}[1]{\Psi(#1)}
\newcommand{\rhobase}{\rho_{0}}

\newcommand{\gzc}[1]{\textcolor{red}{#1}}
\section{Proof of Theorem~\ref{Theorem:Bound}}\label{section:witness}
In this section we fix an $n$-visible lcVASS $\avass=(Q, q_0, \Sigma_n, T, F)$.
Note that since $\avass$ is label-compressed, the effect on the external counters $\exteffof{\rho}$ for a path $\rho$ is also compressed. This matches the effect $\inteffof{\rho}$, which is anyway encoded in binary.

The proof is Rackoff-like, and we explain the analogy as we proceed. 
Like in Rackoff's upper bound for coverability \cite{DBLP:journals/tcs/Rackoff78}, we reason over all (extended) configurations.
Unlike Rackoff, however, we do not look at short covering sequences from a given configuration, but rather at small flowers rooted in said configuration. 
To bound the size while maximizing over all configurations, we measure each flower's size without considering the configuration it is rooted in.  
We therefore define the \emph{flower bound}
\begin{align*}
\flowerbound\ = \ 
\max_{\aconfig}\min \setcond{\sizeof{\sigma}+\sizeof{\abloom}}{(\aconfig.\sigma, \abloom)\text{ is an inseparability flower}}\ .   
\end{align*}
\Cref{Theorem:Bound} is an immediate consequence of the following. 
\begin{lemma}\label{Lemma:FlowerBound}
$\flowerbound \leq 2^{\sizeof{\avass}^{\bigoof{(d+n)^2}}}$. 
\end{lemma}

\subparagraph{Step I: From length bounds to flower size bounds}
Before we prove \cref{Lemma:FlowerBound}, let us see how it implies \cref{Theorem:Bound}. Notice that \cref{Lemma:FlowerBound} makes a statement about the lengths of the runs $\sigma$, $\alpha$, $\beta$, $\gamma$, whereas \cref{Theorem:Bound} also promises a small starting configuration $\aconfig$. Thus, it remains to construct a small starting configuration.
\begin{proof}[Proof of \cref{Theorem:Bound}]
  To begin, assume $\aflower=(\aconfig.\sigma, \abloom)$ is coverable.
  By Lemma~\ref{Lemma:FlowerBound}, there is another flower $(\aconfig.\sigma', \abloom')$ that is rooted in the same extended configuration and satisfies $\sizeof{\sigma'}+\sizeof{\abloom'}\leq \flowerbound$. 
  The extended configuration $\aconfig$ may not obey the desired bound.
  Thus, we replace it by $\aconfig'$ defined by $\omegaof{\aconfig'}=\omegaof{\aconfig}$ and $\coordacc{\aconfig'}{\acounter}=\min\set{\coordacc{\aconfig}{\acounter}, \flowerbound\cdot 2^{\sizeof{\avass}}}$ for all $\acounter\in\counters\setminus\omegaof{\aconfig}$.  
  Here, we keep the $\omega$-entries,
  and for each non-$\omega$-counter, we take the value of $\aconfig$ unless it is larger than $\flowerbound\cdot 2^{\sizeof{\avass}}$, in which case we truncate to this value.
  We now claim that $\aflower'=(\aconfig'.\sigma', \abloom')$ is still a coverable inseparability flower, and its size is at most doubly exponential (satisfying the desired bound). 
  Coverability is immediate by $\aconfig'\leq \aconfig$. 
  We also have $\unarysizeof{\aconfig'}\leq 2^{2^{\bigoof{\sizeof{\avass}}}}$. 
  To be an inseparability flower, we have to check that the counters whose values we truncated when moving from $\aconfig$ to $\aconfig'$ remain non-negative while executing $\sigma$, and also $\sigma.\rho$ with $\rho=\alpha, \beta, \gamma$.
  This, however, is clear by the fact that $\sigma.\rho$ has length at most $\flowerbound$, and each transition can subtract at most $2^{\sizeof{\avass}}$ tokens.
  As we have this budget available in $\aconfig'$, the run remains non-negative.
\end{proof}

In the remainder of the section, we prove Lemma~\ref{Lemma:FlowerBound}. 
Similar to Rackoff's proof for coverability, we generalize the notion of flowers to admit negative counter values. 
Then we use an induction on the number of non-negative counters to establish a bound on the length of shortest generalized flowers. 
Let $\rackcounter\in[0, d+n]$, $\abloom=(\finalstate, \counters, \alpha, \beta, \gamma)$ a bloom, and $\omegacounters=[1, d+n]\setminus\counters$.  
We call an extended run $\aconfig.\sigma$ of $\avass \times \dyckvass{n}$ an \emph{$\rackcounter$-stem for $\abloom$}, if it ends in an extended configuration $\aconfig'$ with $\aconfig'=(\finalstate,\bmm')$ for some $\bmm'$, $\omegacounters\subseteq\omegaof{\aconfig}$, and the counters in $([1, \rackcounter]\cap \counters)\setminus \omegaof{\aconfig}$ remain non-negative when executing $\sigma$ from $\aconfig$ resp.\ $\alpha$, $\beta$, and $\gamma$ from~$\aconfig'$ in $\avass \times \dyckvass{n}$. 
Here, we use $[1, 0]=\emptyset$.
We call a pair $(\aconfig.\sigma, \abloom)$ consisting of an $\rackcounter$-stem and an inseparability bloom an \emph{$\rackcounter$-inseparability flower}.  
Note that a $(d+n)$-inseparability flower is an inseparability flower as %
in \cref{Def:bloom}. 
We say the flower is \emph{$\abound$-bounded} with $\abound\in\nat\setminus\set{0}$, if the counters in $([1, \rackcounter]\cap\counters)\setminus\omegaof{\aconfig}$ are bounded by $\abound$ along $\aconfig.\sigma.\rho$ for all $\rho=\alpha, \beta, \gamma$. 
We wish to establish an estimate on the following function $f$, note that $\rackfunof{d+n} = \flowerbound$: 
\begin{align*}\label{rackfun}
\rackfunof{\rackcounter} = \ 
\max_{\aconfig}\min \setcond{\sizeof{\sigma}+\sizeof{\abloom}}{(\aconfig.\sigma, \abloom)\text{ is an $\rackcounter$-inseparability flower}} \ .
\end{align*}

\subparagraph{Step II: From value bounds to length bounds}
Similar to Rackoff's proof, we will bound $f(i+1)$ in terms of $f(i)$, which will then yield the bound on $\flowerbound=f(n+d)$. However, there is a key difference to Rackoff's proof: When constructing runs that, in the first $i+1$ coordinates stay non-negative, Rackoff argues that if such a run only uses counter values in $[0,b]$ on the first $i+1$ coordinates, then there is such a run of length at most $|Q|\cdot (b+1)^{i+1}$: whenever a combination of values in $[0,b]$ (and a control state) repeats, we can cut out the infix in between. Hence, \emph{value bounds yield length bounds}. Our setting requires a different argument: Simply cutting out infixes in $\alpha,\beta,\gamma$ might spoil the properties (ii), (iii), and (iv) of $i$-inseparability flowers. Instead, we use our results on SNLS to obtain length bounds from value bounds.
The technique is similar to some other generalizations of Rackoff's result by Yen~\cite[Lemma~3.5]{Yen92} Habermehl~\cite[Lemma~3.2]{Habermehl97} and Atig and Habermehl~\cite[Lemma~5]{AtigH11}. However, the following proof also needs to work with non-linear constraints (which are also not present in Demri's extension of Rackoff's result~\cite{Demri13}).

\begin{lemma}\label{Lemma:RBoundedWitnesses}
    Let $\aflower=(\aconfig.\sigma, \abloom)$ be a $\abound$-bounded $\rackcounter$-inseparability flower. 
    Then there is an $\rackcounter$-inseparability flower $\aflower'=(\aconfig.\sigma', \abloom')$ with $\sizeof{\sigma'}+\sizeof{\abloom'}\leq (2^{\sizeof{\avass}}\cdot \sizeof{\states}\cdot \abound)^{\bigoof{(d+n)^{6}}}$. 
\end{lemma}
\begin{proof}
Let $\aflower = (\aconfig.\sigma, \abloom)$ be a $\abound$-bounded $\rackcounter$-inseparability flower with $\abloom =(\finalstate, \counters, \alpha, \beta, \gamma)$. 
Let $\omegacounters$ be the complement of $\counters$. 
Let $\widehat \aconfig$ be the extended configuration reached by $\aconfig.\sigma$.

Let $\distinguished = ([1, \rackcounter]\cap \counters)\setminus\omegaof{\aconfig}$ be the counters we wish to keep non-negative. 
By a \emph{$\distinguished$-loop}, we mean an extended run $\aconfig_1.\tau.\aconfig_2$ of $\avass \times \dyckvass{n}$ where $\tau$ is a loop, $\coordacc{\aconfig_1}{\distinguished}=\coordacc{\aconfig_2}{\distinguished}$, and $\coordacc{\aconfig_1}{\omegacounters}=\coordacc{\aconfig_2}{\omegacounters}$. 
The $\distinguished$-loop is called \emph{irreducible} if it does not contain further $\distinguished$-loops. 
By the pigeonhole principle, any run that is longer than $\dloopbounddef=\sizeof{\states}\cdot (\abound+1)^{\sizeof{\distinguished \cup \omegacounters}}$ contains a $\distinguished$-loop.
This means the length of an irreducible $\distinguished$-loop is at most $\dloopbounddef$. 
Moreover, these loops can have at most $2\cdot 2^{\sizeof{\avass}}\cdot\dloopbounddef$ distinct effects on each counter, where the leading $2$ is for the distinction between positive and negative values. 
This analysis yields an upper bound of $\largebound=(2^{\sizeof{\avass}+1}\cdot\dloopbounddef)^{d+n}$ on the number of irreducible $\distinguished$-loops with distinct effects. 
We show how to construct a flower as promised in the lemma.

We begin by cutting out irreducible $\distinguished$-loops from $\sigma$,
yielding $\sigma'$ of length at most $\dloopbounddef$. 
The extended configuration $\aconfig'$ reached by $\aconfig.\sigma'$ coincides with~$\widehat \aconfig$ on the $\omega$-entries and on the non-$\omega$-entries for $\distinguished$.

Next, we want to shorten $\rho = \alpha, \beta, \gamma$.
To this end, we first decompose them into irreducible $\distinguished$-loops. 
We assume all such $\distinguished$-loops have distinct effects, otherwise we pick a representative. 
With the previous analysis, the result is $\rhobase, \rho_1, \ldots, \rho_{\largeboundp}$ with $\largeboundp\leq \largebound$. 
Here, $\rhobase$ is the loop on the final state into which the irreducible $\distinguished$-loops $\rho_1, \ldots, \rho_{\largeboundp}$ are inserted.  
Note that $\rhobase$ is not necessarily a $\distinguished$-loop.
What we know, however, is that $\rhobase$ has a non-negative effect on the counters in $\counters$, due to Condition~(i) in the definition of inseparability blooms.
Since we want to be able to insert all the $\distinguished$-loops directly into $\rhobase$, the latter should still contain the same starting configurations for such loops as $\rho$, at least when only considering the counters in $\distinguished \cup \omegacounters$.
Therefore we cannot guarantee that $\rhobase$ contains no $\distinguished$-loops, because cutting all of them out might reduce the number of such configurations visited by $\rhobase$.
However, $\rhobase$ still has length at most $\dloopbounddef^2$ by the following argument.
If we mark in $\rhobase$ the first occurrence of each element from $\states \times [0,b]^{\distinguished \cup \omegacounters}$, then each infix leading from one such marker to the next has at most length $\dloopbounddef$, because longer infixes still contain a $\distinguished$-loop that has no marked configuration and can therefore be removed.
Since there are at most $\sizeof{\states \times [0,\abound]^{\distinguished \cup \omegacounters}} = \sizeof{\states}\cdot (\abound+1)^{\sizeof{\distinguished \cup \omegacounters}} = \dloopbounddef$ markers, we obtain the stated length bound of $\dloopbounddef^2$ for $\rhobase$.

A vector $\bx\in\nat^{\largeboundp}$ with $\coordacc{\bx}{0}>0$ and $\coordacc{\bx}{\acounter}\geq 0$ for $1\leq j\leq \largeboundp$ can now be turned into a run $\rho_{\bx}$ by glueing together $\coordacc{\bx}{\acounter}$-many instances of $\rho_{\acounter}$.  
Note that also the base loop $\rhobase$ may be repeated. 
As the order of the transitions does not influence the effect of the run, $\effectof{\rho}=\effectof{\rho_{\by}}$ holds, where $\by=\parikhof{\rho}$ is the so-called Parikh vector that counts the occurrences of irreducible loops in~$\rho$. 
As a consequence, we can directly define the effect on the vector $\bx$, namely $\effectof{\bx}=\sum_{i=0}^{\largeboundp}\coordacc{\bx}{\acounter}\cdot \effectof{\rho_{\acounter}}\in\ints^{d+n}$.  

With irreducible $\distinguished$-loops at hand, we can formulate the search for a small bloom as an SNLS.
We define the vectors $\bx_{\alpha}, \bx_{\beta}, \bx_{\gamma}$, and the following constraints:
\begin{alignat*}{5}
\coordacc{\bx_{\alpha}}{0}, \coordacc{\bx_{\beta}}{0}, \coordacc{\bx_{\gamma}}{0}\geq&\ 1\qquad&\qquad\coordacc{\vassbalanceof{\bx_{\alpha}+\bx_{\beta}+\bx_{\gamma}}}{1, d}\geq&\ 0\\
\bx_{\alpha}, \bx_{\beta},  \bx_{\gamma}\geq&\ 0
&\coordacc{\balanceof{\bx_{\alpha}+\bx_{\beta}}}{d+1, d+n}\geq&\ 0\\ 
\coordacc{\effectof{\bx_{\alpha}}}{\counters}, \coordacc{\effectof{\bx_{\beta}}}{\counters}, \coordacc{\effectof{\bx_{\gamma}}}{\counters}\geq&\ 0&
\coordacc{(\balanceof{\bx_{\alpha}+\bx_{\beta}+\bx_{\gamma}}-\arat\cdot\balanceof{\bx_{\alpha}})}{d+1, d+n}=&\ 0& 
\end{alignat*}

The constraints on the left say that we repeat the base loops at least once, and the remaining loops a non-negative number of times.
The last constraint is Condition (i) in the definition of blooms.
The constraints on the right correspond to the Conditions (ii) to (iv).

It is readily checked that $(\parikhof{\alpha}, \parikhof{\beta}, \parikhof{\gamma})$ solves the SNLS. 
This means Theorem~\ref{solution-size} applies and yields a small rational solution.
To turn it into an integer solution, we observe that our SNLS is monotonic in the sense that if $(\bx_{\alpha}, \bx_{\beta}, \bx_{\gamma})$ is a solution, so is $(k\bx_{\alpha}, k\bx_{\beta}, k\bx_{\gamma})$ for any $k\geq 1$.
We multiply the rational solution by the common denominator to obtain the integer solution $(\bx_{\alpha}', \bx_{\beta}', \bx_{\gamma}')$ with the associated loops $\alpha', \beta', \gamma'$.

We argue that $\aflower'=(\aconfig.\sigma', \abloom')$ with $\abloom'=(\finalstate, \counters, \alpha', \beta', \gamma')$ is an $\rackcounter$-inseparability flower. 
Remember that $\aconfig'$ is the configuration reached by $\aconfig.\sigma'$. 
As $\aconfig$ and $\counters$ have not changed, the only thing we have to show is the non-negativity for the counters in $\distinguished$. 
For $\sigma'$, this follows from the non-negativity of $\sigma$, and the fact that we only cut-out $\distinguished$-loops. 
For $\rho'=\alpha', \beta', \gamma'$, which are executed from $\aconfig'$, we argue as follows.  
The base loop $\rho_0$ has the same effect on the $\distinguished$-counters as $\rho$, because we obtained $\rho_0$ by cutting-out $\distinguished$-loops from~$\rho$. 
The effect of~$\rho$ on even the entire set $\counters$ is non-negative due to Condition~(i) for blooms. 
This means if one repetition of $\rho_0$ stays non-negative from~$\aconfig'$, arbitrarily many repetitions will. 
The one repetition stays non-negative, because $\rho$ was non-negative from $\widehat \aconfig$ (the extended configuration reached by $\aconfig.\sigma$), and $\aconfig'$ coincides with $\widehat \aconfig$ on $\distinguished$. 
Since the $\distinguished$-loops that we glue into $\rhobase$ come with a valuation of the counters in $\distinguished$, and this valuation keeps them non-negative in $\rho$, the entire $\rho'$ stays non-negative.

We analyze the complexity of our system $\cS$.
It has at most $\rowof{\cS}\in \bigoof{d+n}$ many rows, and note that the non-negativity constraints do not count towards the rows. 
There are at most $\colof{\cS}\in \bigoof{\largebound}$ many columns.
The degree is $\degof{\cS}=1$.
The maximal coefficient is bounded from above by the largest possible loop effect, 
$\maxcof{\cS}\leq 2^{\sizeof{\avass}}\cdot\dloopbounddef^2$. 
Then, \cref{solution-size} gives us rational solutions $\bx_{\alpha}, \bx_{\beta}, \bx_{\gamma}$ of the form $\bx_{\rho}[i]=\dfrac{a}{K}$, meaning $K$ is the common denominator of all entries, with 
\begin{align*} 
\unarysizeof{\bx_{\rho}}\in &\   (\colof{\cS}\cdot \degof{\cS} \cdot
\maxcof{\cS})^{\bigoof{\degof{\cS}^2 \cdot \rowof{\cS}^4}}\\
=&\  ( \bigoof{\largebound}\cdot 2^{\sizeof{\avass}}\cdot\dloopbounddef^2)^{\bigoof{(d+n)^4}}\\
=&\  ( \bigoof{(2^{\sizeof{\avass}+1}\cdot \dloopbounddef)^{d+n}}\cdot 2^{\sizeof{\avass}}\cdot\dloopbounddef^2)^{\bigoof{(d+n)^4}}\\
=&\  (2^{\sizeof{\avass}}\cdot \dloopbounddef)^{\bigoof{(d+n)^5}}
=\ (2^{\sizeof{\avass}}\cdot\sizeof{\states}\cdot (\abound+1))^{\bigoof{(d+n)^6}}\\
=&\ (2^{\sizeof{\avass}}\cdot\sizeof{\states}\cdot \abound)^{\bigoof{(d+n)^6}} \ .
\end{align*}

We already argued that the integer vectors $K\bx_{\alpha}$, $K\bx_{\beta}$, $K\bx_{\gamma}$ are also solutions with runs $\alpha', \beta', \gamma'$.  
Since each entry of these vectors is smaller than $(2^{\sizeof{\avass}}\cdot\sizeof{\states}\cdot \abound)^{\bigoof{(d+n)^{6}}}$, and we have at most 
$\largebound = (2^{\sizeof{\avass}}\cdot \sizeof{\states}\cdot (\abound+1))^{\bigoof{(d+n)^2}}$ many loops with maximal size $\dloopbounddef^2 = \sizeof{\states}^2\cdot (\abound+1)^{2(d+n)}$, we get 
    \begin{align*}
        \sizeof{\rho'}&\leq (2^{\sizeof{\avass}}\cdot \sizeof{\states}\cdot
        \abound)^{\bigoof{(d+n)^{6}}}\cdot (2^{\sizeof{\avass}}\cdot
        \sizeof{\states}\cdot(\abound+1))^{\bigoof{(d+n)^2}}\cdot
        \sizeof{\states}^2\cdot (\abound+1)^{2(d+n)}\\
        &=(2^{\sizeof{\avass}}\cdot \sizeof{\states}\cdot \abound)^{\bigoof{(d+n)^{6}}}\ .\qedhere
    \end{align*}
\end{proof}
\subparagraph{Step III: Rackoff-style induction}
We now give the bound on $\rackfunof{\rackcounter}$ that we need to prove \cref{Lemma:FlowerBound}.
In the base case, no counter has to remain non-negative and so we have a $1$-bounded $0$-inseparability flower. 
We employ the bound from Lemma~\ref{Lemma:RBoundedWitnesses}. 
\begin{lemma}
$\rackfunof{0}=(2^{\sizeof{\avass}}\sizeof{\states})^{\bigoof{(d+n)^{6}}}$.
\end{lemma}
In the induction step, and as in Rackoff's result, the bound takes the form of a recurrence.
\begin{lemma}
$\rackfunof{i+1}\leq (2^{\sizeof{\avass}}\cdot \rackfunof{\rackcounter})^{\bigoof{(d+n)^6}}$.
\end{lemma}
\begin{proof}
Consider an $(\rackcounter+1)$-inseparability flower $\aflower=(\aconfig.\sigma, \abloom)$ with $\abloom=(\finalstate, \counters, \alpha, \beta, \gamma)$ and $\omegacounters = [1, d+n]\setminus\counters$. 
Let~$\rackboundof{\rackcounter} = 2^{\sizeof{\avass}}\cdot\rackfunof{\rackcounter}$ serve as an abbreviation. 
We proceed by a case distinction.  
If $\aflower$ is $\rackboundof{\rackcounter}$-bounded, then Lemma~\ref{Lemma:RBoundedWitnesses} provides another $(\rackcounter+1)$-inseparability flower $(\aconfig.\sigma', \abloom')$ with $\sizeof{\sigma'}+\sizeof{\abloom'}\leq (2^{\sizeof{\avass}}\cdot \sizeof{\states}\cdot \rackboundof{\rackcounter})^{\bigoof{(d+n)^{6}}}=(2^{\sizeof{\avass}}\cdot \rackfunof{\rackcounter})^{\bigoof{(d+n)^6}}$.  
This satisfies the bound stated in the lemma.

If $\aflower$ is not $\rackboundof{\rackcounter}$-bounded, then $\sigma.\rho$ with $\rho=\alpha, \beta, \gamma$ exceeds $\rackboundof{i}$.  
We identify the first moment when this happens, say after $\rho_1$ and for the $(i+1)$-th counter. 
The case where the run exceeds the bound already in $\sigma$ is simpler. 
The run decomposes into 
\begin{align*}
\aconfig\xrightarrow{\sigma}\aconfig_1\xrightarrow{\rho_1}\aconfig'\xrightarrow{\rho_2}\aconfig_2\ .
\end{align*}

We argue that also $(\aconfig'.\rho_2, \abloom)$ is an $(\rackcounter+1)$-flower, which means $\aconfig'.\rho_2$ is an $(\rackcounter+1)$-stem for $\abloom$. 
Since $\rho$ is a loop, it returns to~$\finalstate$. 
We have $\omegaof{\aconfig}=\omegaof{\aconfig'}$ and since $\omegacounters\subseteq\omegaof{\aconfig}$, we get $\omegacounters\subseteq\omegaof{\aconfig'}$. 
As $\abloom$ is a bloom, $\rho$ has a non-negative effect on the counters in $\counters$. 
This means $\coordacc{\aconfig_1}{\counters}\leq\coordacc{\aconfig_2}{\counters}$, and so the counters in $([1, \rackcounter+1]\cap\counters)\setminus\omegaof{\aconfig'}$ remain non-negative when executing $\alpha$, $\beta$, $\gamma$ from $\aconfig_2$ as they did from~$\aconfig_1$. 
Also $\rho_2$ remains non-negative from~$\aconfig'$, because $\rho$ remained non-negative from $\aconfig_1$.  

Since $(\aconfig'.\rho_2, \abloom)$ is an $(\rackcounter+1)$-flower, it is an $\rackcounter$-flower.
The induction hypothesis yields another $\rackcounter$-flower $\aflower'=(\aconfig'.\sigma', \abloom')$ with $\sizeof{\sigma'}+\sizeof{\abloom'}\leq \rackfunof{\rackcounter}$. 
Let $\abloom'=(\finalstate', \counters', \alpha', \beta', \gamma')$ and let~$\omegacounters'$ be the complement of $\counters'$.  

We argue that $\aflower'$ is actually an $(\rackcounter+1)$-flower. 
If the counter $i+1$ that exceeds the bound $\rackboundof{\rackcounter}$ does not belong to $\counters'$, there is nothing to show. 
Otherwise, even $\sigma'.\alpha'.\beta'.\gamma'$ in succession could subtract at most $\rackfunof{\rackcounter}\cdot2^{\sizeof{\avass}}=\rackboundof{\rackcounter}$ tokens from counter $i+1$.  
Since this counter carries more than $\rackboundof{\rackcounter}$ tokens, this leaves us with a positive balance.

We show that also $(\aconfig.\sigma.\rho_1.\sigma', \abloom')$ is an $(\rackcounter+1)$-flower, meaning $\aconfig.\sigma.\rho_1.\sigma'$ is an $(\rackcounter+1)$-stem for $\abloom'$. 
We have $\omegaof{\aconfig}=\omegaof{\aconfig'}$ and so $\omegacounters'\subseteq\omegaof{\aconfig'}$ implies $\omegacounters'\subseteq\omegaof{\aconfig}$. 
It remains to show that the counters in $([1, \rackcounter+1]\cap\counters')\setminus \omegaof{\aconfig}$ remain non-negative. 
Consider the prefix $\sigma.\rho_1$ executed from $\aconfig$. 
For the counters that also belong to $\counters$, non-negativity holds as $\aflower$ is an $(i+1)$-flower. 
Assume there was a counter in $([1, \rackcounter+1]\cap\counters')\setminus \omegaof{\aconfig}$ that did not belong to $\counters$. 
Then it belonged to $\omegacounters$.  
But as $\omegacounters\subseteq\omegaof{\aconfig}$, we had a contradiction.
For the suffix $\sigma'$ executed from $\aconfig'$, and for $\alpha', \beta', \gamma'$, non-negativity holds as $\aflower'$ is an $(\rackcounter+1)$-flower. 

To estimate the size of the newly constructed flower rooted in $\aconfig$, we assume $\sigma.\rho_1$ does not repeat configurations on the first $(\rackcounter+1)$-counters. 
If it does, we cut out the infix and adapt the values of the counters that are allowed to fall below zero. 
Then the length of $\sigma.\rho_1$ is bounded by $\sizeof{\states}\cdot\rackboundof{\rackcounter}^{\rackcounter+1}$, and we have 
\begin{align*}
\sizeof{\sigma.\rho_1.\sigma'}+\sizeof{\abloom'}
\leq&\ \sizeof{\states}\cdot\rackboundof{\rackcounter}^{\rackcounter+1}+\rackfunof{\rackcounter} 
\leq\ (2^{\sizeof{\avass}}\cdot \rackfunof{\rackcounter})^{\bigoof{\rackcounter+1}}
\leq (2^{\sizeof{\avass}}\cdot \rackfunof{\rackcounter})^{\bigoof{(d+n)^6}}\ .\qedhere
\end{align*}
\end{proof}

It remains to solve the recurrence.
Let $a=2^{\sizeof{\avass}}$ and $b=\bigoof{(d+n)^6}$. 
We have 
\begin{align*}
\rackfunof{d+n} = (a\ldots (a\cdot \rackfunof{0})^b\ldots )^b \leq (a^{d+n}\cdot \rackfunof{0})^{b^{d+n}}\ .  
\end{align*}
Since $\rackfunof{0}=(2^{\sizeof{\avass}}\cdot\sizeof{\states})^{\bigoof{(d+n)^6}}$, we obtain the promised $\flowerbound=\rackfunof{d+n}\leq 2^{\sizeof{\avass}^{\bigoof{(d+n)^2}}}$.

\label{beforebibliography}
\newoutputstream{pages}
\openoutputfile{main.pages.ctr}{pages}
\addtostream{pages}{\getpagerefnumber{beforebibliography}}
\closeoutputstream{pages}

\bibliographystyle{plainurl}%
\bibliography{bibliography.bib}

\begin{thebibliography}{10}

\bibitem{AtigH11}
M.~F. Atig and P.~Habermehl.
\newblock On {Y}en's path logic for {P}etri nets.
\newblock {\em Int. J. Found. Comput. Sci.}, 22(4):783--799, 2011.
\newblock \href {https://doi.org/10.1142/S0129054111008428}
  {\path{doi:10.1142/S0129054111008428}}.

\bibitem{Baumann23}
P.~Baumann, R.~Meyer, and G.~Zetzsche.
\newblock Regular separability in {B}{\"{u}}chi {VASS}.
\newblock In {\em Proc. STACS}, volume 254 of {\em LIPIcs}, pages 9:1--9:19.
  Schloss Dagstuhl, 2023.
\newblock \href {https://doi.org/10.4230/LIPICS.STACS.2023.9}
  {\path{doi:10.4230/LIPICS.STACS.2023.9}}.

\bibitem{Baumann23full}
Pascal Baumann, Roland Meyer, and Georg Zetzsche.
\newblock Regular separability in {B{\"{u}}chi} {VASS}.
\newblock {\em CoRR}, abs/2301.11242, 2023.
\newblock \href {https://arxiv.org/abs/2301.11242} {\path{arXiv:2301.11242}},
  \href {https://doi.org/10.48550/ARXIV.2301.11242}
  {\path{doi:10.48550/ARXIV.2301.11242}}.

\bibitem{DBLP:conf/mfcs/BlockeletS11}
M.~Blockelet and S.~Schmitz.
\newblock Model checking coverability graphs of vector addition systems.
\newblock In {\em Proc. MFCS}, volume 6907 of {\em LNCS}, pages 108--119.
  Springer, 2011.
\newblock \href {https://doi.org/10.1007/978-3-642-22993-0_13}
  {\path{doi:10.1007/978-3-642-22993-0_13}}.

\bibitem{DBLP:conf/icalp/ClementeCLP17}
L.~Clemente, W.~Czerwi\'{n}ski, S.~Lasota, and C.~Paperman.
\newblock Regular separability of parikh automata.
\newblock In {\em Proc. ICALP}, volume~80 of {\em LIPIcs}, pages 117:1--117:13.
  Schloss Dagstuhl, 2017.
\newblock \href {https://doi.org/10.4230/LIPICS.ICALP.2017.117}
  {\path{doi:10.4230/LIPICS.ICALP.2017.117}}.

\bibitem{DBLP:conf/stacs/ClementeCLP17}
L.~Clemente, W.~Czerwi\'{n}ski, S.~Lasota, and C.~Paperman.
\newblock Separability of reachability sets of vector addition systems.
\newblock In H.~Vollmer and B.~Vall{\'{e}}e, editors, {\em Proc. STACS},
  volume~66 of {\em LIPIcs}, pages 24:1--24:14. Schloss Dagstuhl, 2017.
\newblock \href {https://doi.org/10.4230/LIPICS.STACS.2017.24}
  {\path{doi:10.4230/LIPICS.STACS.2017.24}}.

\bibitem{DBLP:conf/concur/CzerwinskiH22}
W.~Czerwi\'{n}ski and P.~Hofman.
\newblock Language inclusion for boundedly-ambiguous vector addition systems is
  decidable.
\newblock In {\em Proc. CONCUR}, volume 243 of {\em LIPIcs}, pages 16:1--16:22.
  Schloss Dagstuhl, 2022.
\newblock \href {https://doi.org/10.4230/LIPICS.CONCUR.2022.16}
  {\path{doi:10.4230/LIPICS.CONCUR.2022.16}}.

\bibitem{DBLP:conf/icalp/CzerwinskiHZ18}
W.~Czerwi\'{n}ski, P.~Hofman, and G.~Zetzsche.
\newblock Unboundedness problems for languages of vector addition systems.
\newblock In {\em Proc. ICALP}, volume 107 of {\em LIPIcs}, pages
  119:1--119:15. Schloss Dagstuhl, 2018.
\newblock \href {https://doi.org/10.4230/LIPICS.ICALP.2018.119}
  {\path{doi:10.4230/LIPICS.ICALP.2018.119}}.

\bibitem{DBLP:journals/lmcs/CzerwinskiL19}
W.~Czerwi\'{n}ski and S.~Lasota.
\newblock Regular separability of one counter automata.
\newblock {\em Log. Methods Comput. Sci.}, 15(2), 2019.
\newblock \href {https://doi.org/10.23638/LMCS-15(2:20)2019}
  {\path{doi:10.23638/LMCS-15(2:20)2019}}.

\bibitem{DBLP:conf/concur/CzerwinskiLMMKS18}
W.~Czerwi\'{n}ski, S.~Lasota, R.~Meyer, S.~Muskalla, K.~N. Kumar, and
  P.~Saivasan.
\newblock Regular separability of well-structured transition systems.
\newblock In {\em Proc. CONCUR}, volume 118 of {\em LIPIcs}, pages 35:1--35:18.
  Schloss Dagstuhl, 2018.
\newblock \href {https://doi.org/10.4230/LIPICS.CONCUR.2018.35}
  {\path{doi:10.4230/LIPICS.CONCUR.2018.35}}.

\bibitem{DBLP:conf/lics/CzerwinskiZ20}
W.~Czerwi\'{n}ski and G.~Zetzsche.
\newblock An approach to regular separability in vector addition systems.
\newblock In {\em Proc. LICS}, pages 341--354. {ACM}, 2020.
\newblock \href {https://doi.org/10.1145/3373718.3394776}
  {\path{doi:10.1145/3373718.3394776}}.

\bibitem{Demri13}
S.~Demri.
\newblock On selective unboundedness of {VASS}.
\newblock {\em JCSS}, 79(5):689--713, 2013.
\newblock \href {https://doi.org/10.1016/J.JCSS.2013.01.014}
  {\path{doi:10.1016/J.JCSS.2013.01.014}}.

\bibitem{DBLP:journals/jacm/GurariI79}
E.~M. Gurari and O.~H. Ibarra.
\newblock An {NP}-complete number-theoretic problem.
\newblock {\em J. {ACM}}, 26(3):567--581, 1979.
\newblock \href {https://doi.org/10.1145/322139.322152}
  {\path{doi:10.1145/322139.322152}}.

\bibitem{Habermehl97}
P.~Habermehl.
\newblock On the complexity of the linear-time $\mu$-calculus for {{P}etri
  Nets}.
\newblock In {\em Proc. ICATPN}, volume 1248 of {\em LNCS}, pages 102--116.
  Springer, 1997.
\newblock \href {https://doi.org/10.1007/3-540-63139-9_32}
  {\path{doi:10.1007/3-540-63139-9_32}}.

\bibitem{KM69}
R.~M. Karp and R.~E. Miller.
\newblock Parallel program schemata.
\newblock {\em JCSS}, 3(2):147--195, 1969.
\newblock \href {https://doi.org/10.1016/S0022-0000(69)80011-5}
  {\path{doi:10.1016/S0022-0000(69)80011-5}}.

\bibitem{DBLP:conf/concur/Keskin023}
E.~Keskin and R.~Meyer.
\newblock Separability and non-determinizability of {WSTS}.
\newblock In {\em Proc. CONCUR}, volume 279 of {\em LIPIcs}, pages 8:1--8:17.
  Schloss Dagstuhl, 2023.
\newblock \href {https://doi.org/10.4230/LIPICS.CONCUR.2023.8}
  {\path{doi:10.4230/LIPICS.CONCUR.2023.8}}.

\bibitem{KeskinMeyer2024a}
E.~Keskin and R.~Meyer.
\newblock On the separability problem of {VASS} reachability languages.
\newblock In {\em To appear in Proc. of LICS}, 2024.

\bibitem{DBLP:conf/fsttcs/KocherZ23}
C.~K{\"{o}}cher and G.~Zetzsche.
\newblock Regular separators for {VASS} coverability languages.
\newblock In {\em Proc. FSTTCS}, volume 284 of {\em LIPIcs}, pages 15:1--15:19.
  Schloss Dagstuhl, 2023.
\newblock \href {https://doi.org/10.4230/LIPICS.FSTTCS.2023.15}
  {\path{doi:10.4230/LIPICS.FSTTCS.2023.15}}.

\bibitem{DBLP:conf/stoc/Kosaraju82}
S.~R. Kosaraju.
\newblock Decidability of reachability in vector addition systems (preliminary
  version).
\newblock In {\em Proc. STOC}, pages 267--281. {ACM}, 1982.
\newblock \href {https://doi.org/10.1145/800070.802201}
  {\path{doi:10.1145/800070.802201}}.

\bibitem{DBLP:journals/tcs/Lambert92}
J.{-}L. Lambert.
\newblock A structure to decide reachability in {P}etri nets.
\newblock {\em Theor. Comput. Sci.}, 99(1):79--104, 1992.
\newblock \href {https://doi.org/10.1016/0304-3975(92)90173-D}
  {\path{doi:10.1016/0304-3975(92)90173-D}}.

\bibitem{Lang2002}
S.~Lang.
\newblock {\em Algebra, Rev. 3rd Ed.}
\newblock Springer, New York, 2002.

\bibitem{DBLP:conf/popl/Leroux11}
J.~Leroux.
\newblock Vector addition system reachability problem: a short self-contained
  proof.
\newblock In {\em Proc. POPL}, pages 307--316. {ACM}, 2011.
\newblock \href {https://doi.org/10.1145/1926385.1926421}
  {\path{doi:10.1145/1926385.1926421}}.

\bibitem{DBLP:conf/lics/LerouxS15}
J.~Leroux and S.~Schmitz.
\newblock Demystifying reachability in vector addition systems.
\newblock In {\em Proc. LICS}, pages 56--67. {IEEE} Computer Society, 2015.
\newblock \href {https://doi.org/10.1109/LICS.2015.16}
  {\path{doi:10.1109/LICS.2015.16}}.

\bibitem{Lipton76}
R.~J. Lipton.
\newblock The reachability problem requires exponential space.
\newblock Technical Report~63, Yale University, 1976.

\bibitem{macintyre1983elimination}
Angus Macintyre, Kenneth McKenna, and Lou van~den Dries.
\newblock Elimination of quantifiers in algebraic structures.
\newblock {\em Advances in Mathematics}, 47(1):74--87, 1983.
\newblock \href {https://doi.org/10.1016/0001-8708(83)90055-5}
  {\path{doi:10.1016/0001-8708(83)90055-5}}.

\bibitem{Marker2002}
D.~Marker.
\newblock {\em Model Theory: An Introduction}.
\newblock Springer, New York, 2002.

\bibitem{DBLP:journals/siamcomp/Mayr84}
E.~W. Mayr.
\newblock An algorithm for the general {P}etri net reachability problem.
\newblock {\em {SIAM} J. Comput.}, 13(3):441--460, 1984.
\newblock \href {https://doi.org/10.1137/0213029} {\path{doi:10.1137/0213029}}.

\bibitem{Mishra93}
B.~Mishra.
\newblock {\em Algorithmic Algebra}.
\newblock Texts and Monographs in Computer Science. Springer, 1993.
\newblock \href {https://doi.org/10.1007/978-1-4612-4344-1}
  {\path{doi:10.1007/978-1-4612-4344-1}}.

\bibitem{DBLP:journals/tcs/Rackoff78}
C.~Rackoff.
\newblock The covering and boundedness problems for vector addition systems.
\newblock {\em Theor. Comput. Sci.}, 6:223--231, 1978.
\newblock \href {https://doi.org/10.1016/0304-3975(78)90036-1}
  {\path{doi:10.1016/0304-3975(78)90036-1}}.

\bibitem{Savitch70}
W.~J. Savitch.
\newblock Relationships between nondeterministic and deterministic tape
  complexities.
\newblock {\em JCSS}, 4(2):177--192, 1970.
\newblock \href {https://doi.org/10.1016/S0022-0000(70)80006-X}
  {\path{doi:10.1016/S0022-0000(70)80006-X}}.

\bibitem{Schrijver1986}
A.~Schrijver.
\newblock {\em Theory of linear and integer programming}.
\newblock John Wiley \& Sons, 1986.

\bibitem{DBLP:conf/fsttcs/ThinniyamZ19}
R.~S. Thinniyam and G.~Zetzsche.
\newblock Regular separability and intersection emptiness are independent
  problems.
\newblock In {\em Proc. FSTTCS}, volume 150 of {\em LIPIcs}, pages 51:1--51:15.
  Schloss Dagstuhl, 2019.
\newblock \href {https://doi.org/10.4230/LIPICS.FSTTCS.2019.51}
  {\path{doi:10.4230/LIPICS.FSTTCS.2019.51}}.

\bibitem{DBLP:journals/jsc/Weispfenning90}
V.~Weispfenning.
\newblock The complexity of almost linear diophantine problems.
\newblock {\em J. Symb. Comput.}, 10(5):395--404, 1990.
\newblock \href {https://doi.org/10.1016/S0747-7171(08)80051-X}
  {\path{doi:10.1016/S0747-7171(08)80051-X}}.

\bibitem{Yen92}
H.{-}C. Yen.
\newblock A unified approach for deciding the existence of certain {P}etri net
  paths.
\newblock {\em Information and Computation}, 96(1):119--137, 1992.
\newblock \href {https://doi.org/10.1016/0890-5401(92)90059-O}
  {\path{doi:10.1016/0890-5401(92)90059-O}}.

\end{thebibliography}

\newoutputstream{todos}
\openoutputfile{main.todos.ctr}{todos}
\addtostream{todos}{\arabic{@todonotes@numberoftodonotes}}
\closeoutputstream{todos}

\label{endofdocument}
\newoutputstream{pagestotal}
\openoutputfile{main.pagestotal.ctr}{pagestotal}
\addtostream{pagestotal}{\getpagerefnumber{endofdocument}}
\closeoutputstream{pagestotal}

\conferenceFull{}{\newpage}
\conferenceFull{}{\appendix}
\conferenceFull{}{\section{Additional material on Section~\ref{section:quant-elim}}\label{appendix-snls}

We continue our notation $n=\colof{\bA(\param)}$ and $m=\rowof{\bA(\param)}$ from the main paper.
In contexts like $\degof{-}$, $\maxcof{-}$, etc. we ommit the parameter $\param$ where convenient.
For example we write $\degof{\bA}$ instead of $\degof{\bA(\param)}$.

\subsection{Precise bounds for Theorem~\ref{Theorem:QElim}}\label{precise-bounds-qelim}
Here, we prove that the DNFLB $\Phi$ constructed in \cref{phi-qelim} satisfies the bounds promised in \cref{Theorem:QElim}.

To that end, it suffices to compute for all~$R$ the maximal degrees and maximal coefficients of the terms $\detof{\bD_R(\param)}$, $\bD(\param) \cdot \adjof{\bD_R(\param)} \cdot \bc_R(\param)$, and $\detof{\bD_R(\param)} \cdot \bc(\param)$.
Let $d_{R,i,j}(\param)$ be the entry in the $i$-th row and $j$-th column of $\bD_R(\param)$.
Further let $[1,n]=P\uplus Q$ where $Q$ is the set of row indices $i<n$ for which a $j_i<n$ exists with $d_{R, i, j_i}=1$ and $d_{R, i, j}=0$ for all $j\neq j_i$.
We construct $\bD_R(\param)$ by first augmenting $\bA$ by $n$ rows of the form $(\ldots, 1, \ldots)$ with exactly one non-zero component, and then choosing $\sizeof{R}=n$ rows.
This implies $\sizeof{P}\leq m$.
We recall Leibniz' formula for computing the determinant $\detof{\bD_R(\param)} = \sum_{\sigma \in S_n} \sgnof{\sigma}\cdot \prod_{i = 1}^n d_{R,i,\sigma(i)}$. 
Define $T=\setcond{\sigma\in S_n}{\forall i\in Q.\; \sigma(i)=j_i}$, where $j_i$ is the non-zero valued column of the $i$-th row in $\bD_{R}$.
We observe that if $\sigma \in S_n$ with $\sigma(i)\neq j_{i'}$ for some $i'\in Q$, then $\prod_{i = 1}^n d_{R,i,\sigma(i)}=0$ because it contains the operand $d_{R, i', j_{i'}}=0$.
Then, 
$$\sum_{\sigma \in S_n} \sgnof{\sigma}\cdot \prod_{i = 1}^n d_{R,i,\sigma(i)}=\sum_{\sigma \in T}\sgnof{\sigma}\cdot \prod_{i = 1}^n d_{R,i,\sigma(i)}=\sum_{\sigma \in T}\sgnof{\sigma}\cdot \prod_{i \in P} d_{R,i,\sigma(i)}.$$
The latter equality follows from the fact that $\sigma(i)=j_i$ for all $\sigma\in T$ and $i\in Q$, where we have $d_{R, i, j_i}=1$ by definition of $Q$.
The formula shows $\degof{\detof{\bD_R(\param)}} \leq \sizeof{P} \cdot \degof{\cS'}\leq m \cdot \degof{\cS'}$. 
Because $\sigma(i)=j_i$ for all $\sigma\in S_n$ and $i\in Q$, this fixes the images of all but $\sizeof{P}$ inputs.
Since $\sigma\in T$ are permutations, we can conclude $\sizeof{T}=\sizeof{P}!\leq m!$.
Regarding the maximal coefficient, note that each polynomial $d_{R,i,j}(\param)$ is an additive term of at most $\degof{\bD_R(\param)} + 1$ monomials, each with a potentially large coefficient.
Therefore, Leibniz' formula and the above equality yield
  \begin{align*}
    \maxcof{\detof{\bD_R(\param)}}~ &\leq m!\cdot\big((\degof{\bD_R(\param)}+1) \cdot \maxcof{\bD_R(\param)}\big)^{\sizeof{P}} \\
    &\leq m!\cdot\big((\degof{\bD_R(\param)}+1) \cdot \maxcof{\bD_R(\param)}\big)^{m} \\
    &\leq 2^{m \log_2 m} \cdot (\degof{\cS'}+1)^m \cdot (\maxcof{\cS'})^{m}\\
    &\leq (\degof{\cS'} \cdot \maxcof{\cS'})^{\bigoof{m^2}}.
  \end{align*}
Recall that each entry of the adjugate $\adjof{\bD_R(\param)}$ can be computed by deleting one row and one column from $\bD_R(\param)$ and computing the determinant of the resulting matrix.
This implies $\degof{\adjof{\bD_R(\param)}}\leq\degof{\detof{\bD_R(\param)}}$.
We cannot conclude $\maxcof{\adjof{\bD_R(\param)}} \leq \maxcof{\detof{\bD_R(\param)}}$, because terms with conflicting signs in the Leibniz formula might negate each other in $\detof{\bD_R(\param)}$, while not doing so in $\adjof{\bD_R(\param)}$.
However, the approximation we do above for $\maxcof{\detof{\bD_R(\param)}}$ assumes that all summands have the same sign, and therefore is also an upperbound for $\maxcof{\adjof{\bD_R(\param)}}$.
We have 
  \begin{align*}
    \degof{\bD(\param) \cdot \adjof{\bD_R(\param)} \cdot \bc_R(\param)}~
    &\leq 2\cdot \degof{\cS'} + m \cdot \degof{\cS'}\\
    &= (m + 2) \cdot \degof{\cS'}\\
    &= \bigoof{m \cdot \degof{\cS'}}\ . 
  \end{align*}
For the maximal coefficient, we again remark that each polynomial consists of at most $\deg + 1$ monomials,
and further that computing each entry per matrix multiplication here involves adding up $n$ terms.

  Thus, we get
  \begin{align*}
    &\maxcof{\bD(\param) \cdot \adjof{\bD_R(\param)} \cdot \bc_R(\param)} \\
    &~\begin{aligned}
      \leq~ & (\degof{\bD(\param)} \cdot \maxcof{\bD(\param)}) \cdot n \cdot (\degof{\adjof{\bD_R(\param)}} \cdot \maxcof{\adjof{\bD_R(\param)}}) \\
      &\cdot n \cdot (\degof{\bc_R(\param)} \cdot \maxcof{\bc_R(\param)}) \\
      \leq~ & (\degof{\cS'} \cdot \maxcof{\cS'})^{2} 
      \cdot n^{2} \cdot (m \cdot \degof{\cS'})
      \cdot (\degof{\cS'} \cdot \maxcof{\cS'})^{\bigoof{m^2}}\\
      =~ &  (n\cdot \degof{\cS'} \cdot \maxcof{\cS'})^{\bigoof{m^2}}\ .
    \end{aligned}
  \end{align*}
  $\degof{\detof{\bD_R(\param)} \cdot \bc(\param)} \leq  m \cdot \degof{\cS'} + \degof{\cS'} = \bigoof{m \cdot \degof{\cS'}}$,
  \begin{align*}
    &\maxcof{\detof{\bD_R(\param)} \cdot \bc(\param)}\\
    &~\leq (\maxcof{\detof{\bD_R(\param)}}\cdot \degof{\detof{\bD_R(\param)}})\cdot m\cdot(\maxcof{\bc(\param)}\cdot\degof{\bc(\param)})\\
    &~\leq (\degof{\cS'} \cdot \maxcof{\cS'})^{\bigoof{m^2}}\cdot \degof{\cS'}^{2}
    \cdot \maxcof{\cS'} \\
    &~= (\degof{\cS'} \cdot \maxcof{\cS'})^{\bigoof{m^2}}\ .
  \end{align*}
  Since $\degof{\cS'} = \degof{\cS}$ and $\maxcof{\cS'} = \maxcof{\cS}$ by construction,
  we obtain the required values for polynomials appearing in $\qelimformdef$.

\subsection{Precise bounds for Lemma~\ref{exist-solutions-size}}\label{appendix-exist-solutions-size}
In this subsection, we provide a proof of \cref{exist-solutions-size} that yields the precise bounds promised in the statement.

With a DNFLB $\qelimformdef$ in hand, we now establish a bound on the size of a smallest solution.  
To this end we apply a root separation theorem that gives a lower bound on the distance of distinct real roots of a polynomial. 
We define the \emph{real roots} of a polynomial $\apoly(\param)\in\ints[\param]$ as  
$\realrootsof{\apoly(\param)} = \set{-\infty,\infty} \cup \setcond{\areal \in \R}{\apoly(\areal) = 0}$. 
The advantage of including $-\infty$ and $\infty$ is that the entire set $\R$ decomposes into intervals between adjacent roots of $\apoly(\param)$, which will save us case distinctions. 
Thus, recall that the number of (even complex) roots of a polynomial of degree~$d$ is at most $d+2$. 
For a DNFLB $\qelimformdef$, we define $\realrootsof{\qelimformdef}$ as the set of all real roots among its polynomials. 
Finally, we define $\Separation{\apoly(\param)} = \min_{r \neq r' \in \realrootsof{\apoly(\param)}}\sizeof{r-r'}$ as the minimal distance between two distinct real roots. 
We restate Rump's bound from the main paper.
\rump*

We now formulate the feasibility of $\qelimformdef$ in
terms of constraints on real roots.
Consider one of its lower bound constraints $\apoly(\param) > 0$.  
Let $-\infty = r_1 < r_2 < \ldots < r_k = \infty$ be the real roots of $\apoly(\param)$.
The key observation is that, on the entire interval $(r_i, r_{i+1})$, the polynomial $\apoly(\param)$ is either positive or negative. 
We evaluate up-front which is the case, and collect in the set $P\subseteq \realrootsof{\apoly(\param)}$ the roots $r_i$ where a positive interval starts.
Then the lower bound constraint $\apoly(\param) > 0$ is equivalent to $\bigvee_{r_i\in P} r_i<\param<r_{i+1}$. 
We refer to a formula of the form $\ell<\param<u$ as an \emph{interval constraint}, and call $\ell$ a lower and $u$ an upper bound. 

Consider a conjunction of interval constraints $\bigwedge_{i\in I}\ell_i < \param< u_i$ with $I$ a finite index set. 
A rational number $\arat\in\rats$ will only satisfy this conjunction if it is larger than the largest lower bound $\ell=\max_{i\in I}\ell_i$, and smaller than the least upper bound $u=\min_{i\in I}u_i$. 
This means the conjunction is equivalent to the single interval constraint $\ell<\param<u$. 

To deal with lower bound constraints $\apoly(\param) \geq 0$ that are not strict, we
generalize interval constraints to admit both relations, $<$ and $\leq$. 
To generalize the elimination of conjunctions, these relations can also be mixed, so we also consider $\ell\leq \param< u$ an interval constraint.
We refer to a disjunction of interval constraints as a DINC $\dincformdef$. 
We use $\Bounds{\dincformdef}$ for the set of bounds that appear in $\dincformdef$. 
The discussion yields the following.

\begin{lemma}\label{Lemma:DINC}
For every DNFLB~$\qelimformdef$, there is a DINC $\dincformdef$ that satisfies $\qelimformdef\ratequiv\dincformdef$ and $\Bounds{\dincformdef}\subseteq \realrootsof{\qelimformdef}$. 
\end{lemma}

With this, we are finally ready to prove \cref{exist-solutions-size}.

\begin{proof}[Proof of \cref{exist-solutions-size}]
Consider an SNLS $\cS(\param, \by) = \bA(\param)\cdot \by \geq \bb(\param) \wedge \by \geq 0$  that is feasible (over the rationals). 
Let $\bA(\param)$ be an $m \times n$ matrix. 
We first apply \cref{Theorem:QElim} to obtain the DNFLB $\qelimformdef$ whose polynomials have a degree bounded by $\bigoof{m \cdot \degof{\cS}}$ and a maximal coefficient bounded by $n^{2}\cdot (\degof{\cS} \cdot \maxcof{\cS})^{\bigoof{m^2}}$. 
We then apply \cref{Lemma:DINC} to turn this DNFLB into a DINC $\dincformdef$. 

Since the SNLS is feasible and the two transformations yield formulae that are equivalent over the rationals, there is an interval constraint $\ell<\param<u$ in $\dincformdef$ that has a rational solution.
We show that there is even a rational $\arat=\frac{a}{b}$ with $\ell<\arat<u$
that satisfies the bound promised by \cref{exist-solutions-size}. 
We proceed by a case distinction along $\ell$ and $u$. 

\begin{description}
    \item[Case $\bm{\ell < u}$:]
    In the special case $\ell = -\infty$ and $u = \infty$, we pick $a = 0$ and $b = 1$ as a solution.
    Assume $u \neq \infty$, the case $\ell \neq -\infty$ is similar.
    Then $\ell\in \realrootsof{\apoly_1}$ and $u \in \realrootsof{\apoly_2}$ for two polynomials in $\qelimformdef$.
    Note that $\apoly_1$ and $\apoly_2$ may be different.
    However, $\ell$ and $u$ are both roots of the product polynomial $\apoly_\times = \apoly_1\cdot\apoly_2$, meaning 
    $\ell, u\in \realrootsof{\apoly_\times}$. 
    Note that $\degof{\apoly_\times} = \degof{\apoly_1} +\degof{\apoly_2}\in\bigoof{m \cdot \degof{\cS}}$ and
    $\unarysizeof{\apoly_\times} \leq \unarysizeof{\apoly_1} \cdot \unarysizeof{\apoly_2}$.
    Therefore using Rump's Bound we get
    \begin{align*}
      & \sizeof{u - \ell}\\
      ~>~& \big(\degof{\apoly_\times}^{\degof{\apoly_\times}+1}
      \cdot (1+\sizeof{\apoly_\times}_1)^{2\degof{\apoly_\times}}\big)^{-1}\\
      ~\geq~& \left(\begin{aligned}
      &\bigoof{m\cdot\degof{\cS}}^{\bigoof{m\cdot\degof{\cS}}} \\
      &\cdot \Big(\bigoof{m\cdot\degof{\cS}} \cdot (n\cdot \degof{\cS} \cdot \maxcof{\cS})^{\bigoof{m^2}}\Big)^{\bigoof{m\cdot\degof{\cS}}}
      \end{aligned}\right)^{-1}\\
      ~\geq~& (n\cdot \degof{\cS} \cdot \maxcof{\cS})^{-\bigoof{m^3\cdot\degof{\cS}}}.
    \end{align*}
    In the following, we use $B$ as a shorthand for
    $(n\cdot \degof{\cS} \cdot \maxcof{\cS})^{\bigoof{m^3\cdot\degof{\cS}}}$.
    Immediately, we get $B(u - \ell) > 1$.
    Thus, there is an integer $a \in \Z$ such that $B \ell < B u - 1 \leq a < B u$.
    Dividing by $B$ yields $\ell < \frac{a}{B} < u$, meaning we can choose a
    denominator $b \in \N \setminus \{0\}$ for a suitable solution with $b \leq B$. 
    Let $u$ be a root of $\apoly_1$. 
    We apply Cauchy's bound Lemma~\ref{cauchy-bound} on root sizes,
    and get
    \begin{equation*}
      \sizeof{u} \leq 1 + \sizeof{\apoly_1}_1 \leq \bigoof{m \cdot \degof{\cS}} \cdot (n\cdot \degof{\cS} \cdot \maxcof{\cS})^{\bigoof{m^2}}
      \leq (n\cdot \degof{\cS} \cdot \maxcof{\cS})^{\bigoof{m^2}},
    \end{equation*}
    and therefore $\sizeof{a} \leq B\sizeof{u} + 1 \leq B \cdot (n\cdot \degof{\cS} \cdot \maxcof{\cS})^{\bigoof{m^2}}$.
    Due to the $\mathcal{O}$-notation, the second factor is subsumed by $B$ here.
  \item[Case $\bm{\ell = u}$:]
    Then $\frac{a}{b} = \ell = u$, and this is a root of a polynomial $\apoly(\param)$ in $\qelimformdef$.
    Thus, there is a polynomial $p$ in $\qelimformdef$ such that $\frac{a}{b} \in \Roots{p}$.
    Since also $\frac{a}{b} \in \Q$ we get that $a,b \leq \maxc(p)$
    using the Rational Root Theorem (Theorem~\ref{rational-root-theorem}).
    This implies $a,b \leq (n\cdot \degof{\cS} \cdot \maxcof{\cS})^{\bigoof{m^2}}$.
  \end{description}
  In all cases $B = (n\cdot \degof{\cS} \cdot \maxcof{\cS})^{\bigoof{m^3\cdot\degof{\cS}}}$
  is a bound on the complexity of $\param$, as required.
\end{proof}

\subsection{Proof of Theorem~\ref{solution-size}}\label{appendix-solution-size}
We can now put the ingredients together and show \cref{solution-size}. 
We write $n=\colof{\cS}=\colof{\cS'}$ and $m=\rowof{\cS}$, as in the previous subsections.
According to \cref{exist-solutions-size}, we know that if our system
\[ \cS := \bA(\param)\cdot \by \geq \bb(\param) \wedge \by \geq 0\]
is feasible with some $(t,\bs)\in\Q\times\Q^n$, then for some $t\in\Q$ with 
\[\unarysizeof{t}\le (n\cdot \degof{\cS} \cdot \maxcof{\cS})^{\bigoof{m^3\cdot\degof{\cS}}}. \]
We can now apply the argument from \cref{Theorem:QElim} again: We have shown there that if there exists a solution $\bs$ for a particular $t$ with $\bA(t)\cdot\bs\geq \bb(t)$ and $\bs\geq\bzero$, then for some $R\subseteq[1,m+n]$, for which $\bD_R(t)$ is invertible, the following $\bs^*$ is a solution to $\cS$:
\[ \bs^* = \bD_R^{-1}(t)\cdot c_R(t) = \frac{\adj(\bD_R(t))}{\det(\bD_R(t))}\cdot c_R(t). \]
    
Moreover, in the proof of \cref{Theorem:QElim}, we have seen that
\begin{align*}
	\deg(\det(\bD_R(x))) & \le m\cdot \deg(\cS') \\
	\deg(\adj(\bD_R(x))) & \le m\cdot \deg(\cS') \\
	\maxc(\det(\bD_R(x))) & \le (n\cdot \deg(\cS')\cdot\maxc(\cS'))^{\bigoof{m^2}} \\
	\maxc(\adj(\bD_R(x))) & \le (n\cdot \deg(\cS')\cdot\maxc(\cS'))^{\bigoof{m^2}} 
\end{align*}
\begin{align*}
	\unarysizeof{\adj(D_R(t))}&\le \deg(\adj(\bD_R(x)))\cdot \maxc(\adj(\bD_R(x)))\cdot\unarysizeof{t}^{\deg(\adj(\bD_R(\param)))} \\
	&\le m\cdot\deg(\cS')\cdot (n\cdot \deg(\cS')\cdot\maxc(\cS'))^{\bigoof{m^2}}\cdot \unarysizeof{t}^{m\cdot\deg(\cS')} \\
  &\le (n\cdot \deg(\cS')\cdot\maxc(\cS'))^{\bigoof{m^2}}\cdot \unarysizeof{t}^{m\cdot\deg(\cS')} \\
	&\le (n\cdot \deg(\cS')\cdot\maxc(\cS'))^{\bigoof{m^2}}\cdot \left((n\cdot \degof{\cS}\cdot\maxc\cS)^{\bigoof{m^3\cdot\degof{\cS}}}\right)^{m\cdot\deg(\cS')} \\
	&\le (n\cdot \degof{\cS}\cdot\maxcof{\cS})^{\bigoof{m^4\cdot\degof{\cS}^2}}
\end{align*}

Since for $\det(D_R(x))$, we have the same bounds on maximal coefficients and degree as for $\adj(D_R(x))$, we can show similarly:
\begin{align*}
	\unarysizeof{\det(\bD_R(t))}\le  (n\cdot\degof{\cS}\cdot\maxcof{\cS})^{\bigoof{m^4\cdot\degof{\cS}^2}}
\end{align*}
Moreover, clearly:
\[ \unarysizeof{c_R(t)}\le \maxcof{\cS}\cdot \unarysizeof{t}\le (n\cdot\degof{\cS} \cdot \maxc\cS)^{\bigoof{m^3\cdot\degof{\cS}}}. \]
Furthermore, we have
\[ \unarysizeof{\adj(\bD_R(t))\cdot c_R(t)} \le n\cdot (\unarysizeof{\adj(\bD_R(t))}+\unarysizeof{c_R(t)})\le (n\cdot\degof{\cS}\cdot\maxcof{\cS})^{\bigoof{m^4\cdot\degof{\cS}^2}}. \]
Therefore, we can upper bound $\bs^*$ as follows:
\begin{align*}
	\unarysizeof{\bs^*}&=\left\|\frac{\adj(\bD_R(t))}{\det(\bD_R(t))}\cdot c_R(t)\right\|_1\le \unarysizeof{\adj(\bD_R(t))\cdot c_R(t)}+\unarysizeof{\det(\bD_R(t))} \\
	&\le (n\cdot \degof{\cS}\cdot\maxcof{\cS})^{\bigoof{m^4\cdot\degof{\cS}^2}} + (n\cdot \degof{\cS}\cdot\maxcof{\cS})^{\bigoof{m^4\cdot\degof{\cS}^2}} \\
	&\le (n\cdot \degof{\cS}\cdot\maxcof{\cS})^{\bigoof{m^4\cdot\degof{\cS}^2}} \\
	&\le (\colof{\cS}\cdot \degof{\cS}\cdot\maxcof{\cS})^{\bigoof{\rowof{\cS}^4\cdot\degof{\cS}^2}}
\end{align*}
As we mentioned in the main paper, if we can ensure that all components of $c_{R}(t)$ have the same denominator, we also ensure this for $\bs^{*}$.
To avoid clutter, we did not argue this property in the bound above.
However, this does not change the bound.
In order to ensure that the components of $c_{R}(t)$ all have the same denominator, we only need to multiply the denominator and nominator of each component by a factor of at most $q^{\degof{\cS}}$ where $q$ is the denominator of $t$.
This is because the components of $c_{R}(t)$ are obtained by summing integer multiples of $1, t^{1}, \ldots, t^{\degof{\cS}}$.
The bit size of $c_{R}(t)$, and thus the resulting $\bs^{*}$, blows up by a factor of at most $\degof{\cS}$, which is absorbed by the term $(\colof{\cS}\cdot \degof{\cS}\cdot\maxcof{\cS})^{\bigoof{\rowof{\cS}^4\cdot\degof{\cS}^2}}$.
This completes the proof of \cref{solution-size}.

}
\conferenceFull{}{
\section{Additional material on Section~\ref{section:search}}\label{appendix:search}
In this section we show Lemma~\ref{Lemma:InsepCharacterization}.
To do this, we use the inseparability characterization given in 
\cite{Baumann23}.
We break Lemma~\ref{Lemma:InsepCharacterization} down into its two directions.
\begin{lemma}\label{Lemma:CharacterizationBck}
  Let $\avass$ be an $n$-visible VASS.
  If there is a coverable inseparability flower $(\aconfig.\sigma, \abloom)$, then $\notregsep{\langof{\avass}}{D_{n}}$.
\end{lemma}
\begin{lemma}\label{Lemma:CharacterizationFwd}
  Let $\avass$ be an $n$-visible VASS.
  If $\notregsep{\langof{\avass}}{D_{n}}$, then there is a coverable inseparability flower $(\aconfig.\sigma, \abloom)$.
\end{lemma}
For Lemma~\ref{Lemma:CharacterizationBck}, we use the 
basic separator characterization of regular separability in \cite{Baumann23}.
Given a B\"uchi VASS $\avass$, regular separability of its language $\langof{\avass}$ from $D_{n}$ is characterized in terms of inclusion in a finite union of so-called basic separators.
The authors identify two families of languages that together make up the set of basic separators, 
$\setcond{P_{i, k}}{i \in [1,n], k\in \nat}$ and $\setcond{S_{\aweight, k}}{\bx\in\N^{n}, k\in\nat}$.
We will define these families below. Let us first formally state the characterization.
\begin{lemma}[\hspace{1sp}\cite{Baumann23}]\label{Lemma:BasicSeparators}
    Let $\avass$ be an $n$-visible VASS.
    Then, $\regsep{\langof{\avass}}{D_{n}}$ if and only if there is a number $k\in\nat$ and a finite set of vectors $X\subseteq\nat^{n}$ such that 
    \[\langof{\avass}\subseteq\bigcup_{i \in [1,n]} P_{i, k}\cup\bigcup_{\aweight \in X}S_{\aweight, k}\ .\]
\end{lemma}
Both language families impose conditions on the letter balance along a word.
Given a word $w\in\Sigma_{n}^{*}$ its \emph{letter balance} $\varphi(w)$ is a vector in $\Z^n$, whose $i$-th component $\varphi_i(w)$ counts how many more letters $a_i$ than $\bar{a}_i$ appear in $w$.
In other words this balance is equal to the effect on the external counters of a finite run $\aconfig.\apath$ that is labeled by the word $w$, i.e.\ $\varphi(w) = \exteffof{\apath}$.

The language $P_{i, k} \subseteq \Sigma_{n}^\omega$ contains words whose balance is bounded by $k$ in the $i$-th component before becoming negative for some prefix.
In contrast, the language $S_{\aweight, k} \subseteq \Sigma_{n}^\omega$ is prefix independent.
It contains words whose $\aweight$-weighted balance gradually decreases beyond some point, while all suffixes beyond the same point have an $\aweight$-weighted balance bounded by $k$.
Here the \emph{$\aweight$-weighted balance} of a word $v \in \Sigma_n^*$ for some vector $\aweight \in \nat^n$ is the dot product $\dotpof{\aweight}{\varphi(v)} \in \Z$.

These language families are defined formally below, for all $i \in [1,n]$, $\aweight\in\nat^{n}$, and $k\in\nat$.
For a word $\aword\in\Sigma_{n}^{*}$ or $\aword\in\Sigma_{n}^{\omega}$, we use $\prefixof{\aword}$ to denote the set of finite prefixes of $\aword$, and $\infixof{\aword}$ to denote the set of finite infixes of $\aword$.
\begin{align*}
    P_{i, k}&=\setcond{\aword\in\Sigma_{n}^{\omega}}{\exists\awordp\in\prefixof{\aword}.\ibalanceof{i}{\awordp}<0\wedge\forall\awordpp\in\prefixof{\awordp}.\ibalanceof{i}{\awordpp}\leq k}\\
    S_{\aweight, k}&=\setcond{\aword.\awordp\in\Sigma_{n}^{\omega}}{\forall\awordpp\in\infixof{\awordp}.\dotpof{\aweight}{\balanceof{\awordpp}}\leq k \;\;\wedge\;\;
    \awordp=\awordp_{0}.\awordp_{1}\ldots, \forall j\in\nat.\;\dotpof{\aweight}{\balanceof{\awordp_{j}}}<0}
\end{align*}
With this at hand, we show Lemma~\ref{Lemma:CharacterizationBck}.
We recall the argument in \cite[Proposition 5.5] {Baumann23} and extend it to include the basic separators of the form $P_{i, k}$ for all $i \in [1,n]$, $k\in\nat$.
\begin{proof}[Proof of Lemma~\ref{Lemma:CharacterizationBck}]
    Let $\avass$ be an $n$-visible VASS, and let $(\aconfig.\sigma, \abloom)$ be a coverable inseparability flower where $\abloom=(\finalstate, \counters, \alpha, \beta, \gamma)$ and $[1, d+n]=\Omega\uplus\counters$.
    Towards a contradiction, suppose $\regsep{\langof{\avass}}{D_{n}}$.
    Then, by Lemma~\ref{Lemma:BasicSeparators}, we know that there is a number $k\in\nat$ and a finite set $X\subseteq\nat^{n}$ such that
    \[\langof{\avass}\subseteq \bigcup_{i\in[1,n]}P_{i,k}\cup\bigcup_{\aweight\in X}S_{\aweight, k} \ .\]
    Let the effect $\allbalanceof{-}$ of the prefixes of $\alpha$, $\beta$, and $\gamma$ be bounded from below by $-y\leq -1$ for every counter, and let the run $\aconfig.\sigma$ end in the extended configuration $(q_f, \bmm)$.
    Let $\bmm'\in\nat^{d+n}$ be obtained from $\bmm$ by replacing all instances of $\omega$ with $3\cdot (k+1)\cdot y$.
    Since $(\aconfig.\sigma, \abloom)$ is coverable, $(q_f, \bmm)$ is coverable as well, hence there must be witnessing runs covering the non-$\omega$ counters of $(q_f, \bmm)$ and reaching arbitrarily high values for the counters in $\omega(q_f, \bmm)\supseteq\Omega$.
    In particular, there must be some run $(q_0, \zerovec).\apath$ that covers the configuration $(q_f,\bmm')$.

    We know that upon taking $\alpha$,  $\beta$, and $\gamma$ from $(\finalstate, \bmm)$, the counter $i$ remains positive and the counter valuation does not decrease, for all $i\in \counters$.
    Clearly this remains true if we repeat these loops $(k+1)$-times, i.e.\ taking $\alpha^{k+1}$, $\beta^{k+1}$, and $\gamma^{k+1}$.
    By our choice of $y$, we also know that starting from some counter valuation $u\in\nat$ for counter $i\in[1, d+n]$, during execution of $\alpha$, $\beta$, and $\gamma$ we always remain above the valuation $u-y$ for this counter.
    Therefore, since $\bmm'$ has a counter valuation of $3\cdot (k+1)\cdot y$ for all counters in $\Omega \subseteq \omega(q_f,\bmm)$, taking $\alpha^{k+1}.\beta^{k+1}.\gamma^{k+1}$ from $(q_f,\bmm')$ keeps the counters in $\Omega$ positive as well.
    So, we deduce that the run $(q_0, \zerovec).\apath.\alpha^{k+1}.\beta^{k+1}.\gamma^{k+1}$ remains positive in all counters $[1, d+n]$.
    Furthermore condition (ii) for an inseparability bloom, $\vassbalanceof{\alpha}+\vassbalanceof{\beta}+\vassbalanceof{\gamma}\geq 0$, implies that $\alpha^{k+1}.\beta^{k+1}.\gamma^{k+1}$ can be infinitely repeated to produce a run that remains positive in the counters $\set{1, \ldots d}$.
    Because $\alpha$, $\beta$, and $\gamma$ are loops on the final state $\finalstate$, the run $(q_0, \zerovec).\apath.(\alpha^{k+1}.\beta^{k+1}.\gamma^{k+1})^{\omega}$ is an accepting run in $\avass$.
    Let $\aword \in \Sigma_n^\omega$ be the word labeling this run.
    The argument in the proof of \cite[Proposition 5.5] {Baumann23} applies to show $\aword\not\in S_{\aweight, k}$ for all $\aweight\in X$.

    It remains to argue that $\aword\not\in P_{i, k}$ for all $i<n$.
    If $d+i\in\counters$, then we know that the counter $d+i$ never becomes negative along the run, and thus the same holds for the $i$th component of the letter balance.
    Then $\aword\not\in P_{i, k}$.
    If $d+i\in \Omega$, then taking $\apath$ from $(q_0, \zerovec)$ covers $(q_f, \bmm')$, where $\bmm'[d+i]\geq k+1$ by definition.
    Therefore the $i$th letter balance exceeds bound $k+1$ before it can drop below $0$, so $\aword\not\in P_{i, k}$.
\end{proof}

For the direction of Lemma~\ref{Lemma:CharacterizationFwd}, we do not directly argue from an assumption of inseparability.
Instead, we use another characterization provided in \cite[Theorem 5.3]{Baumann23}.
This characterization follows in two steps.
First, it is shown that a B\"{u}chi VASS $\avass_{\mathsf{pump}}$ can be constructed from $\avass$ such that the language $\langof{\avass_{\mathsf{pump}}}$ is (so-called) \emph{pumpable}, and $\regsep{\langof{\avass}}{D_{n}}$ if and only if $\regsep{\langof{\avass_{\mathsf{pump}}}}{D_{n}}$.
The definition of pumpability is not needed for the purposes of this paper.

For our proof, we only need the second step of the characterization: If $\langof{\avass}$ is pumpable, then $\notregsep{\langof{\avass}}{D_{n}}$ if and only if the B\"uchi Automaton $\mathsf{KM}(\avass)$ contains a structure called a \emph{KM-inseparability flower}\footnote{In \cite{Baumann23} this structure is simply called \emph{inseparability flower}. Since we already use this term for a different structure, which does not involve the KM-graph, we felt it appropriate to add the prefix ``KM-'' for the version from \cite{Baumann23}, to differentiate the two.}.
We define these kinds of flowers and the notation $\mathsf{KM}(\avass)$ below.
Putting the two steps together, we get the following characterization.
\begin{lemma}[\hspace{1sp}\cite{Baumann23}]\label{Lemma:OldKMCharacterization}
    Let $\avass$ be an $n$-visible VASS.
    Then, $\notregsep{\langof{\avass}}{D_{n}}$ if and only if $\mathsf{KM}(\avass_{\mathsf{pump}})$ contains a KM-inseparability flower.
\end{lemma}
We recall the definitions we need.
For the rest of this section, fix a B\"{u}chi VASS $\avass=(\states, q_0, \Sigma_{n}, T, \finalstates)$.
We start with the definition of the Karp-Miller graph \cite{KM69} of a VASS, as formulated in \cite{Baumann23}.
The \emph{Karp-Miller graph} $\mathsf{KM}(\avass)=(\states_\mathsf{KM}, (q_0, \zerovec), T, T_{\mathsf{KM}}, \finalstates\times\nat_{\omega}^{d})$ is a B\"uchi Automaton with states $\states_\mathsf{KM} \subseteq \states\times\nat_{\omega}^{d}$, whose transitions are labeled with the transitions of $\avass$.
Its states and transitions are constructed by starting with $\states_\mathsf{KM} = \emptyset$ and $T_{\mathsf{KM}} = \set{(q_0, \zerovec)}$, and repeating the following steps towards a fixed point.
Pick a transition $t=(q, a, \bx ,q')\in T$ and a state $(q, \bmm) \in T_{\mathsf{KM}}$.
First simulate $(q, a, \bx, q')$ from $(q, \bmm)$, to get $(q', \bmm + \bx)$.
If $(q', \bmm + \bx)\notin\states\times\nat^{d+n}_{\omega}$, continue.
In the other case construct $\bmm_{\omega}$ from $\bmm + \bx$ as follows:
(1) for $i\in[1,d]$, if there is a state $(q', \bmm'), \in \states_\mathsf{KM}$ with $\bmm'\leq \bmm + \bx$ and $\bmm'[i]<(\bmm+\bx)[i]$ such that $(q, \bmm)$ is reachable from $(q', \bmm')$ using the transitions generated so far, then set $\bmm_{\omega}[i]=\omega$;
(2) otherwise, set $\bmm_{\omega}[i]=(\bmm+\bx)[i]$.
Finally set $\states_\mathsf{KM} := \states_\mathsf{KM} \cup \set{(q',\bmm_\omega)}$ and $T_{\mathsf{KM}} := T_{\mathsf{KM}} \cup \set{(q, \bmm) \xrightarrow{t} (q',\bmm_\omega)}$.
This process reaches a fixed point after finitely many states, since $(\nat_{\omega}^{d}, \leq)$ is a well quasi order.
In particular, this means that the state space of the Karp Miller graph is finite.
It is immediate from the definition that if $(q, \bmm)\in\states\times\nat^{d}_{\omega}$ can be reached in $\mathsf{KM}(\avass)$, then $(q, \bmm)$ is coverable in $\avass$.

Now we define $\avass_{\mathsf{pump}}$.
The VASS $\avass_{\mathsf{pump}}=(\states_{\mathsf{pump}}, q_0, \Sigma_{n}, T_{\mathsf{pump}}, \finalstates\times\nat^{d+n}_{\omega})$ has the states $\states_{\mathsf{KM}}$ of the Karp-Miller graph $\mathsf{KM}(\avass\times \mathcal{D}_{n})$ as its state space, and $((q, \bmm), a, \bx, (q', \bmm'))\in T_{\mathsf{pump}}$ if and only if $(q, \bmm) \xrightarrow{t} (q', \bmm')$ is a transition in $\mathsf{KM}(\avass\times \mathcal{D}_{n})$ for some $t=(q, a, \bx, q')$.

A KM-inseparability flower in $\mathsf{KM}(\avass_{\mathsf{pump}})$ consists of a reachable final state $(q_f, \bmm)$ in $\mathsf{KM}(\avass_{\mathsf{pump}})$ with three loops $\alpha$, $\beta$, $\gamma$ on that state that fulfill the following three conditions.
\begin{enumerate}[(a)]
    \item $\delta(\alpha)+\delta(\beta)+\delta(\gamma)\geq\zerovec$,
    \item $\varphi(\alpha)+\varphi(\beta)\geq\zerovec$,
    \item there is $t\in\rats$ with $\varphi(\alpha)+\varphi(\beta)+\varphi(\gamma)-t\cdot\varphi(\alpha)=\zerovec$.
\end{enumerate}
Here, for a path (or cycle) $\rho$ in the Karp-Miller graph, we write $\delta(\rho)$ to refer to the effect of the transition sequence it was labelled by.
In other words for all $\chi \in \set{\alpha, \beta, \gamma}$ in $\mathsf{KM}(\avass_{\mathsf{pump}})$ labelled by a sequence of transitions $\rho_\chi$ of $\avass_{\mathsf{pump}}$, we define $\delta(\chi) := \delta(\rho_\chi)$.

We are now ready to prove Lemma~\ref{Lemma:CharacterizationFwd}.

\begin{proof}[Proof of Lemma~\ref{Lemma:CharacterizationFwd}]
    Let $\avass$ be an $n$-visible VASS with $\notregsep{\langof{\avass}}{D_{n}}$.
    Then by \cref{Lemma:OldKMCharacterization} $\mathsf{KM}(\avass_{\mathsf{pump}})$ contains a KM-inseparability flower.
    This means that there is a state $(\finalstate, \bmm)$ of $\mathsf{KM}(\avass_{\mathsf{pump}})$ and there are three loops $\alpha$, $\beta$, $\gamma$ on this state with the properties (a), (b), (c).
    Recall that the state space of $\avass_{\mathsf{pump}}$ is that of $\mathsf{KM}(\avass\times \mathcal{D}_{n})$.
    Also recall that $((q, \bmn), a, \bx, (q', \bmn'))\in T_{\mathsf{pump}}$ if and only if there is a transition $(q, \bmn) \xrightarrow{t} (q', \bmn')$ in $\mathsf{KM}(\avass\times \mathcal{D}_{n})$ for some $t=(q, a, \bx, q') \in T$.
    Therefore $q_{f}=(q_f', \bmm')$ is a state in $\mathsf{KM}(\avass\times \mathcal{D}_{n})$, for some $\bmm'\in\nat^{d+n}_{\omega}$ and final $q_f'$ in $\avass$ with three loops $\alpha'$, $\beta'$, $\gamma'$ on it in $\mathsf{KM}(\avass\times \mathcal{D}_{n})$ that fulfill (a), (b), and (c).
    Here $\alpha'$, $\beta'$, $\gamma'$ are obtained by extracting the $\avass_{\mathsf{pump}}$-transitions from the labels of $\alpha$, $\beta$, $\gamma$ and then matching them to the corresponding transitions in $\mathsf{KM}(\avass\times \mathcal{D}_{n})$.
    Since $\alpha'$, $\beta'$, $\gamma'$ are themselves labelled by transition sequences (of $\avass\times\mathcal{D}_{n}$), we can extract again to obtain the loops $\alpha_{\avass}$, $\beta_{\avass}$, $\gamma_{\avass}$ in $\avass$ (ignoring the $\mathcal{D}_{n}$-components of the transitions).
    We claim that $((q_f', \bmm').\varepsilon, \abloom)$ is an inseparability flower with $\abloom=(q_f', \counters, \alpha_\avass, \beta_\avass, \gamma_\avass)$ and $\counters=\setcond{i\in \set{1, \ldots d+n}}{\bmm'[i]\neq \omega}$.
    Since $(q_f', \bmm')$ is reachable in $\mathsf{KM}(\avass\times \mathcal{D}_{n})$, we get that $(q_f', \bmm')$ is coverable in $\avass\times \mathcal{D}_{n}$.
    This shows the coverability of $((q_f', \bmm').\varepsilon, \abloom)$.
    For all $i\in \counters$, the loops $\alpha'$, $\beta'$, $\gamma'$ lead from a non-$\omega$ counter valuation to the same non-$\omega$ counter valuation in $\mathsf{KM}(\avass\times \mathcal{D}_{n})$ at counter $i$.
    Since non-$\omega$ counter valuations are precisely tracked in Karp-Miller graphs, we deduce that $\allbalanceof{\rho_\avass}[\counters]=\zerovec$ for all $\rho_\avass\in\set{\alpha_\avass, \beta_\avass, \gamma_\avass}$.
    This shows Condition (i) of inseparability flowers.
    The Conditions (ii), (iii), (iv) are the same as (a), (b), (c), respectively, and are inherited from $\alpha'$, $\beta'$, $\gamma'$.
    It only remains to show that $(q_f', \bmm').\varepsilon$ is a stem for $\abloom$.
    The state $q_f'$ is already final in $\avass$ so $(q_f', \bmm').\varepsilon$ trivially ends in a final state.
    Since $\bmm'\in\nat^{d+n}_{\omega}$ the counters trivially remain non-negative along $(q_f', \bmm').\varepsilon$.
    We also know $\Omega=\set{1, \ldots, d+n}\setminus\counters=\omega(q_f', \bmm')$.
    The trivial run $(q_f', \bmm').\varepsilon$ reaches $(q_f', \bmm')$.
    The argument for $\alpha$, $\beta$, $\gamma$ remaining non-negative on the counters $\counters$ when executed from $(q_f', \bmm')$ is similar to the argument for condition (i).
    The counters $i\in\counters$ are non-$\omega$ and thus precisely tracked along $\alpha$, $\beta$, $\gamma$ from $(q_f', \bmm')$.
    Since the Karp-Miller graph does not allow transitions into or from negative counter valuations, we deduce that the valuations of the counters $\counters$ remain positive along $\alpha$, $\beta$, $\gamma$ from $(q_f', \bmm')$.
    This concludes the proof.
\end{proof}

}

\end{document}